\documentclass{llncs}
\usepackage{amssymb}
\usepackage{stmaryrd}
\newcommand{\keywords}[1]{\par\addvspace\baselineskip
\noindent\keywordname\enspace\ignorespaces#1}
\usepackage{amstext}
\usepackage{a4,a4wide}
\usepackage{amsmath}
\usepackage{url}
\begin{document}

\title{Every Formula-Based Logic Program Has a Least Infinite-Valued Model\footnote{This paper appears in the Proceedings of the 19th International Conference on Applications of Declarative Programming and Knowledge Management (INAP 2011).}}
\titlerunning{Every Formula-Based Logic Program Has a Least Infinite-Valued Model}
\author{Rainer L\"udecke}
\authorrunning{Rainer L\"udecke}
\institute{Universit\"at T\"ubingen, Wilhelm-Schickard-Institut,\\
Sand 13, 72076 T\"ubingen, Germany\\
\email{luedecke@informatik.uni-tuebingen.de}}
\toctitle{The Existence of Least Models for Formula-Based Logic Programs}
\tocauthor{Rainer L\"udecke}
\maketitle
\begin{abstract}
Every definite logic program has as its meaning a least Herbrand model with respect to the program-independent ordering $\subseteq$. In the case of normal logic programs there do not exist least
models in general. However, according to a recent approach by Rondogiannis and Wadge, who consider infinite-valued models,
every normal logic program does have a least model with respect to a
program-independent ordering. We show that this approach can be extended
to formula-based logic programs (i.e., finite sets of rules of the form $A\leftarrow \phi$ where $A$ is an atom and $\phi$ an arbitrary first-order formula).
We construct for a given program $P$ an interpretation $M_P$ and show that
it is the least of all models of $P$.
\keywords{Logic programming, semantics of programs, negation-as-failure, infinite-valued logics, set theory}
\end{abstract}

\section{Introduction}
It is well-known that every definite logic program $P$ has a Herbrand model and the intersection of all its Herbrand models is also a model of $P$. We call it the least Herbrand model or the canonical model of $P$ and constitute that it is the intended meaning of the program. If we consider a normal logic program $P$ it is more complicated to state the intended meaning of the program because the intersection of all its models is not necessarily a model. There are many approaches to overcome that problem. The existing approaches are not purely model-theoretic (i.e., there are normal logic programs that have the same models but different intended meanings). However, there is a recent purely model-theoretic approach of P. Rondogiannis and W. Wadge \cite{RondogiannisWadge}. They prove that every normal logic program has a least infinite-valued model. Their work is based on an infinite set of truth values, ordered as follows:$$\mathcal{F}_0<\mathcal{F}_1<...<\mathcal{F}_\alpha<...<0<...<\mathcal{T}_\alpha<...<\mathcal{T}_1<\mathcal{T}_0$$Intuitively, $\mathcal{F}_0$ and $\mathcal{T}_0$ are the classical truth values \textit{False} and \textit{True}, $0$ is the truth value \textit{Undefined} and $\alpha$ is an arbitrary countable ordinal. The considered ordering of the interpretations is a program-independent ordering on the infinite-valued interpretations and generalizes the classical ordering on the Herbrand interpretations. The intended meaning of a normal logic program is, as in the classical case, stated as the unique minimal infinite-valued model of $P$. Furthermore, they show that the 3-valued interpretation that results from the least infinite-valued model of $P$ by  collapsing all true values to True and all false values to False coincides with the well-founded model of $P$ introduced in \cite{Przymusinski}.
\\Inspired by \cite{Schroeder-Heister} we consider in this paper formula-based logic programs. A formula-based logic program is a finite set of rules of the form $A\leftarrow \phi$, where $A$ is an atomic formula and $\phi$ is an arbitrary first-order formula. We show that the construction methods to obtain the least infinite-valued model of a normal logic program $P$ given in  \cite{RondogiannisWadge} can be adapted to  formula-based logic programs. The initial step to  carry out this  adaption is the proof of two extension theorems.  Informally speaking, these theorems state  that a complex  formula shows the same behavior as  an atomic formula.  While Rondogiannis and Wadge \cite{RondogiannisWadge} make use of the fact that the bodies of normal program rules are conjunctions of negative or positive atoms, we instead make use of one of the extension theorems. The second step to achieve the adaption is the set-theoretical fact that the least uncountable cardinal $\aleph_1$ is  regular (i.e., the limit of a countable sequence of countable ordinals is in $\aleph_1$). Contrary to the bodies of normal program rules, the bodies of formula-based program rules can refer a ground atom to a countably infinite set of ground atoms. This is the reason why we must use in our approach $\aleph_1$ many iteration steps in the construction of the least model of a given program $P$ in conjunction with the regularity of $\aleph_1$. In \cite{RondogiannisWadge} $\omega$ many iteration steps in conjunction with the fact that the limit of a finite sequence of natural numbers is once again a natural number is sufficient to construct the least model. Towards the end of the paper, we use again the regularity of $\aleph_1$ to show that there is a countable ordinal $\delta_{\textup{max}}$ with the property that every least model of a formula-based logic-program refers only to truth values of the form $\mathcal{T}_\alpha$ or $\mathcal{F}_\alpha$ or $0$, where $\alpha<\delta_{\textup{max}}$. This implies that we only need a very small fragment of the truth values if we consider the meaning of a formula-based logic program. Finally, we  show that the 3-valued interpretation that results from the least infinite-valued model of a given formula-based logic program by  collapsing all true values to True and all false values to False, is  a model of $P$ in the sense of \cite{Przymusinski}. But compared to the case of normal logic programs, the collapsed least infinite-valued model of a formula-based logic program is not a minimal 3-valued model of $P$ in general. However, there is a simple restriction for the class of formula-based programs such that the collapsed model is minimal in general.\\At this point we would like to mention that we did not develop the theory presented in this paper with respect to applied logic. We have  a predominantly theoretical interest in extending the notion of \textit{inductive definition} to a wider class of rules. \\We make heavy use of ordinal numbers in this paper. Therefore, we included an appendix with a short introduction to ordinal numbers for those readers who are not familiar with this part of set theory. Moreover, one can find the omitted proofs and a detailed discussion of an example within the appendix. It is downloadable at:\\ \url{http://www-ls.informatik.uni-tuebingen.de/luedecke/luedecke.html}         
  
\section{Infinite-Valued Models}We are interested in logic programs based on a first-order Language $\mathcal{L}$ with finitely many predicate symbols, function symbols, and constants. 
\begin{definition}\label{Sprache Def}\textup{The alphabet of $\mathcal{L}$ consists of the following symbols, where the numbers $n$, $m$, $l$, $s_1$,...,$s_n$, $r_1$,...,$r_m$ are natural numbers such that $n,l,r_i\geq1$ and $m,s_i\geq0$ hold:\begin{enumerate}
\item Predicate symbols: $P_{1},...,P_n$ with assigned arity $s_{1},...,s_{n}$\item Function symbols: $f_1,...,f_m$ with assigned arity $r_{1},...,r_{m}$\item Constants (abbr.: $\textup{Con}$): $c_{1},...,c_{l}$\item Variables (abbr.: $\textup{Var}$): $x_{k}\text{ provided that }k\in\mathbb{N}$\item Connectives:
$\wedge,\vee,\neg,\forall,\exists,\perp,\top$\item Punctuation symbols: '(', ')' and ','\end{enumerate}}\end{definition}
The natural numbers $n,m,l,s_1,...,s_n,r_1,...,r_m\in\mathbb{N}$ are fixed and the language $\mathcal{L}$ only depends on these numbers. If we consider different languages of this type, we will write $\mathcal{L}_{n,m,l,(s_i),(r_i)}$ instead of $\mathcal{L}$ to prevent confusion. The following definitions depend on $\mathcal{L}$. However, to improve readability, we will not mention this again. 
\begin{definition}\textup{
The set of \textit{terms} $\textup{Term}$ is the smallest one satisfying: \begin{enumerate}
\item Constants and variables are in Term.
\item \text{If $t_1,...,t_{r_i}\in \textup{Term}$, then $f_i(t_1,...,t_{r_i})\in \textup{Term}$}.
\end{enumerate}
}\end{definition}
\begin{definition}\textup{The \textit{Herbrand universe} $H_U$ is the set of ground terms (i.e., terms that contain no variables).}
\end{definition}
\begin{definition}\textup{The set of \textit{formulas} $\textup{Form}$ is the smallest one satisfying:
\begin{enumerate}
\item \text{$\perp$ and $\top$ are elements of Form.}
\item \text{If }$t_1,...,t_{s_k}\in \textup{Term}\text{, then }P_k (t_1,...,t_{s_k})\in \textup{Form}.$
\item \text{If }$\phi,\psi\in \textup{Form}$\text{ and }$v\in \textup{Var}$\text{, then }$\neg(\phi),\,(\phi\wedge\psi),\,(\phi\vee\psi),\,\forall v(\phi),\,\exists v(\phi)\in \textup{Form.}$
\end{enumerate}}\end{definition}
 An \textit{atom} is a formula only constructed by means of 1.) or 2.) and a \textit{ground atom} is an atom that contains no variables.
\begin{definition}\textup{The \textit{Herbrand base} $H_B$ is the set of all ground atoms, except $\perp$ and $\top$.}\end{definition}
\begin{definition}\textup{A \textit{(formula-based)} \textit{rule} is of the form $A\leftarrow \phi$ where $\phi$ is an arbitrary formula and $A$ is an arbitrary atom provided that $A\neq\top$ and $A\neq\perp$.\\A \textit{(formula-based) logic program} is a finite set of (formula-based) rules. Notice that we write $A\leftarrow$ instead of $A\leftarrow\top$. Remember $A\leftarrow \phi$ is called a \textit{normal rule} (resp. \textit{definite rule}) if $\phi$ is a conjunction of literals (resp. positive literals). A finite set of normal (resp. \textup{definite}) rules is a \textit{normal (resp. definite) program}. 
}\end{definition}
\begin{definition}\textup{The set of truth values $W$ is given by $$W:=\{\left\langle 0,n \right\rangle;\,\,n\in\aleph_1\}\cup\left\{ 0 \right\}\cup\left\{ \left\langle 1,n\right\rangle;\,\,n\in\aleph_1\right\}.$$ Additionally, we define a  strict linear ordering $<$ on $W$ as follows:\begin{enumerate}\item $\left\langle  0,n\right\rangle <0$ and $0<\left\langle  1,n\right\rangle $ for all $n\in\aleph_1$
\item $\left\langle w,x \right\rangle<\left\langle y,z \right\rangle \quad\text{iff }\\\quad(w=0=y \text{ and }x\in z)\text{ or }(w=1=y\text{ and }z\in x)\text{ or }(w=0\text{ and }y=1)$
\end{enumerate} We define $\mathcal{F}_i:=\left\langle 0,i \right\rangle$ and $\mathcal{T}_i:=\left\langle 1,i \right\rangle$ for all $i\in\aleph_1$. $\mathcal{F}_i$ is a \textit{false} value and $\mathcal{T}_i$ is a \textit{true} value. The value $0$ is the \textit{undefined} value. The following summarizes the situation ($i\in\aleph_1$):\begin{center}$\mathcal{F}_0<\mathcal{F}_1<\mathcal{F}_2<...<\mathcal{F}_i<...<0<...<\mathcal{T}_i<...<\mathcal{T}_2<\mathcal{T}_1<\mathcal{T}_0$\end{center}
}\end{definition}
\begin{definition}\textup{The \textit{degree} (abbr.: $\textup{deg}$) of a truth value is given by $\textup{deg}(0):=\infty$, $\text{deg}(\mathcal{F}_\alpha):=\alpha$, and $\text{deg}(\mathcal{T}_\alpha):=\alpha$ for all $\alpha\in\aleph_1$.}\end{definition}
\begin{definition}\textup{An \textit{(infinite-valued Herbrand) }\emph{interpretation} $I$ is a function from the Herbrand base $H_B$ to the set of truth values $W$. A \textit{variable assignment} $h$ is a mapping from $\textup{Var}$ to $H_U$.
}\end{definition}
\begin{definition}\textup{Let $I$ be an interpretation and $w\in W$ be a truth value, then $I\Vert w$ is defined as the inverse image of $w$ under $I$ (i.e., $I\Vert w=\{A\in H_B;\;I(A)=w\}$).}\end{definition}
\begin{definition}\textup{Let $I$ and $J$  be interpretations and $\alpha\in\aleph_1$. We write $I=_\alpha J,$ if for all $\beta\leq \alpha$, $I\Vert \mathcal{F}_\beta=J\Vert \mathcal{F}_\beta$ and $I\Vert \mathcal{T}_\beta=J\Vert \mathcal{T}_\beta$.}\end{definition}
\begin{definition}\textup{Let $I$ and $J$  be interpretations and $\alpha\in\aleph_1$. We write  $I\sqsubseteq_\alpha J$, if for all $\beta< \alpha$, $I=_\beta J$ and furthermore $J\Vert \mathcal{F}_\alpha\subseteq I\Vert \mathcal{F}_\alpha\;\&\;I\Vert \mathcal{T}_\alpha\subseteq J\Vert \mathcal{T}_\alpha$. We write $I\sqsubset_\alpha J$, if $I\sqsubseteq_\alpha J$ and $I\neq_\alpha J$.}\end{definition}
Now we define a partial ordering $\sqsubseteq_\infty$ on the set of all  interpretations. It is easy to see that this ordering generalizes the classical partial ordering  $\subseteq$ on the set of 2-valued Herbrand interpretations. 
\begin{definition}\textup{Let $I$ and $J$  be interpretations. We write $I\sqsubset_\infty J$, if there exists an $\alpha\in\aleph_1$  such that $I\sqsubset_\alpha J$. We write $I\sqsubseteq_\infty J$, if $I\sqsubset_\infty J$ or $I=J$.}\end{definition}
\begin{remark} To motivate these definitions let us briefly recall the classical 2-valued situation. Therefore let us pick two (2-valued) Herbrand interpretations $I, J\subseteq H_B$. Considering these, it becomes apparent that $I\subseteq J$ holds if and only if the set of ground atoms that are false w.r.t. $J$ is a subset of the set of ground atoms that are false w.r.t. $I$ and the set of ground atoms that are true w.r.t. $I$ is a subset of the set of ground
atoms that are true w.r.t. $J$.\end{remark}
\begin{definition}\textup{Let $h$ be a variable assignment. The \textit{semantics of terms} is given by (with respect to $h$):
\begin{enumerate}
\item $\llbracket c\rrbracket_h=c$ if $c$ is a constant. 
\item $\llbracket v\rrbracket_h=h(v)$ if $v$ is a variable.
\item $\llbracket f_i(t_1,...,t_{r_i})\rrbracket_h=f_i(\llbracket t_1\rrbracket_h,...,\llbracket t_{r_i}\rrbracket_h)$ if $1\leq i\leq m$ and $t_1,...,t_{r_i}\in \textup{Term}$.
\end{enumerate}}\end{definition}
Before we start to talk about the semantics of formulas, we have to show that every subset of $W$ has a \textit{least upper bound} (abbr: $\textup{sup}$) and a \textit{greatest lower bound} (abbr: $\textup{inf}$). The proof of the following lemma is left to the reader. The proof is using the fact that every nonempty subset of $\aleph_1$ has a least element.  
\begin{lemma}\label{sup inf existenz annahme} For every subset $M\subseteq W$ the least upper bound $\textup{sup}\,M$ and the greatest lower bound $\textup{inf}\,M$ exist in $W$. Moreover, $\textup{sup}\,M\in\{\mathcal{T}_\alpha;\,\alpha\in\aleph_1\}$ implies that $\textup{sup}\,M\in M$ and  on the other hand $\textup{inf}\,M\in\{\mathcal{F}_\alpha;\,\alpha\in\aleph_1\}$ implies that $\textup{inf}\,M\in M$.
\end{lemma}

\begin{definition}\label{wfunktion}\textup{Let $I$ be an interpretation and $h$   be a variable assignment. The \textit{semantics of formulas} is given by (with respect to $I$ and $h$):\begin{enumerate}\item If $t_1,...,t_{s_k}\in  \textup{Term}$, then $\llbracket P_k (t_1,...,t_{s_k})\rrbracket^I_h=I\left(P_{k}(\llbracket t_1\rrbracket_{h},...,\llbracket t_{s_k}\rrbracket_h)\right)$. Additionally, the semantics of $\top$ and $\bot$ is given by $\llbracket\top \rrbracket^I_h=\mathcal{T}_0$ and $\llbracket \bot\rrbracket^I_h=\mathcal{F}_{0}$.
\item If $\phi,\psi\in \textup{Form}$ and $v$ an arbitrary variable, then $\llbracket \phi\wedge\psi\rrbracket^I_h=\textup{min}\{\llbracket \phi\rrbracket^I_h,\llbracket \psi\rrbracket^I_h\}$, $\llbracket\phi\vee\psi\rrbracket^I_h=\textup{max}\{\;\llbracket \phi\rrbracket^I_h\;,\;\llbracket \psi\rrbracket^I_h\;\}$, $\llbracket\exists v(\phi)\rrbracket^I_h=\textup{sup}\{\llbracket\phi\rrbracket^I_{h[v\mapsto u]};\;u\in H_U\}$, $\llbracket\forall v(\phi)\rrbracket^I_h=\textup{inf}\{\llbracket\phi\rrbracket^I_{h[v\mapsto u]};\;u\in H_U\}$ and $\llbracket \neg(\phi)\rrbracket^I_h=\begin{cases} 
\mathcal{T}_{\alpha+1}, &\text{if }\llbracket\phi\rrbracket^I_h=\mathcal{F}_{\alpha}\\
\mathcal{F}_{\alpha+1}, &\text{if }\llbracket\phi\rrbracket^I_h=\mathcal{T}_{\alpha}\\
0, &\text{otherwise} \end{cases}.$ 
\end{enumerate}
}\end{definition}

\begin{definition}\label{defmodel}\textup{Let $A\leftarrow \phi$ be a rule, $P$ a program and $I$ an interpretation. Then $I$ \textit{satisfies} $A\leftarrow \phi$ if for all variable assignment $h$ the property $\llbracket A\rrbracket^I_h\geq\llbracket \phi\rrbracket_{h}^I$ holds. Furthermore, $I$ is a \textit{model} of $P$ if $I$ satisfies all rules of $P$.  
}\end{definition}
\begin{definition}\textup{Let $A\leftarrow \phi$  be a rule and $\sigma$ be a variable substitution (i.e., a function from $\textup{Var}$ to $\textup{Term}$  with finite support). Then, $A\sigma\leftarrow \phi\sigma$ is a \textit{ground instance} of the rule $A\leftarrow \phi$ if $A\sigma\in H_B$ and all variables in $\phi\sigma$ are in the scope of a quantifier. It is easy to see that that $\llbracket A\sigma\rrbracket^I_h$ and $\llbracket \phi\sigma\rrbracket^I_{h}$ (with respect to an interpretation $I$ and a variable assignment $h$) depend only on $I$. That is why we write also $\llbracket A\sigma\rrbracket^I$ and $\llbracket \phi\sigma\rrbracket^I$. We denote the \textit{set of all ground instances} of a program $P$ with $P_G$.
}\end{definition}
\begin{example} Consider the formula-based program $P$ given by the set of rules $\{P(c)\leftarrow,\; R(x)\leftarrow\neg P(x),\;P(Sx)\leftarrow \neg R(x),\;Q\leftarrow\forall x\left(P(x)\right)\}$. Then it is easy to prove that the Herbrand interpretation $I=\{P(S^nc)\mapsto \mathcal{T}_{2n};\; n\in\mathbb{N}\}\cup\{R(S^nc)\mapsto \mathcal{F}_{2n+1};\;n\in\mathbb{N}\}\cup\{Q\mapsto \mathcal{T}_\omega\}$ is a model of $P$. Moreover, using the results of this paper one can show that it is also the least Herbrand model of $P$. 
\end{example}
\begin{remark} Before we proceed we want to give a short informal but intuitive description of the semantics given above. Let us consider two rabbits named Bugs Bunny and Roger Rabbit. We know about them, that Bugs Bunny is a grey rabbit and if Roger Rabbit is not a grey rabbit, then  he is a white one. This information can be understood as a normal logic program: \begin{center}\textit{grey(Bugs Bunny)}$\;\Leftarrow$\\\textit{white(Roger Rabbit)}$\;\Leftarrow\;$\textit{not grey(Roger Rabbit)}\end{center} There is no doubt that \textit{Bugs Bunny is grey} is true because it is a fact. There is also no doubt that every try to prove that \textit{Roger Rabbit is grey} will fail. Hence, using the negation-as-failure rule, we can infer that \textit{Roger Rabbit is white} is also true. But everybody would agree that there is a difference of quality between the two statements because negation-as-failure is not a sound inference rule. The  approach of \cite{RondogiannisWadge}  suggests that the ground atom \textit{grey(Bugs Bunny)}  receives the best possible truth value named $\mathcal{T}_0$ because it is a fact of the program. The atom  \textit{grey(Roger Rabbit)}  receives the worst possible truth value named $\mathcal{F}_0$ because of the negation-as-failure approach. Hence, using the above semantics for negation, \textit{white(Roger Rabbit)}  receives only the second best truth value $\mathcal{T}_1$. 
\end{remark}
\section{The Immediate Consequence Operator} 
\begin{definition}\label{def immediate consequence}\textup{Let $P$ be a program, then the \textit{immediate consequence operator} $T_P$  for  the program $P$ is a mapping from and into $\{I;\;I\text{ is an interpretation}\}$,
where $T_P(I)$ maps an $A\in H_B$ to $T_P(I)(A):=\textup{sup}\{\llbracket \phi\rrbracket^I;\text{ $A\leftarrow \phi\in P_G$}\}.$ (Notice that $P_G$ can be infinite and hence we cannot use max instead of sup.)
}\end{definition}
\begin{definition}\textup{ Let $\alpha$ be an arbitrary countable ordinal. A function $T$ from and into the set of interpretations is called \textit{$\alpha$-monotonic} iff for all interpretations $I$ and $J$ the property $I\sqsubseteq_{\alpha}J\Rightarrow T(I)\sqsubseteq_{\alpha}T(J)$ holds. }\end{definition}
We will show that $T_P$ is $\alpha$-monotonic. Before we will  give the proof of this property, we have to prove the first extension theorem.
\begin{theorem}[Extension Theorem I]\label{Fortsetzungssatz1} Let $\alpha$ be an arbitrary  countable ordinal and $I$, $J$  two interpretations provided that $I\sqsubseteq_{\alpha}J$.  The following properties hold for every formula $\phi$:\begin{enumerate}\item If $\mathcal{F}_0\leq w\,\leq \mathcal{F}_{\alpha}$ and $h$ an assignment, then $\llbracket \phi\rrbracket^J_h=w\:\Rightarrow\:\llbracket \phi \rrbracket^I_h=w$.\item If $\mathcal{T}_{\alpha}\;\leq w\leq \mathcal{T}_{0}$ and $h$ an assignment, then $\,\llbracket\phi \rrbracket^I_h=w\;\Rightarrow\;\llbracket \phi \rrbracket^J_h=w$.\item If $\textup{deg}(w)\,\,<\alpha\;\,$ and $h$ an assignment, then $\,\llbracket\phi \rrbracket^I_h=w\;\Leftrightarrow\;\llbracket \phi \rrbracket^J_h=w$.\end{enumerate}\end{theorem}
\begin{proof}We show these statements by induction on $\phi$. Let $I_{H}(X)$ be an abbreviation for 1. and 2. and 3., where $\phi$ is replaced by $X$ (induction hypothesis).\\
\textit{Case 1:} $\phi=\top$ or $\phi=\perp$. In this case 1., 2., and 3. are obviously true.
\\\textit{Case 2:} $\phi=P_k(t_1,...,t_{s_k})$. 1., 2., and 3. follow directly from $I\sqsubseteq_{\alpha}J$.
\\\textit{Case 3:} $\phi=\neg(A)$. We assume that $I_H(A)$. 
We show simultaneously that 1., 2. and 3. also hold. Therefore, we choose an assignment $h$ and a truth value $w$ such that $\mathcal{F}_0\leq w\leq \mathcal{F}_{\alpha}$ resp. $\mathcal{T}_{\alpha}\leq w\leq \mathcal{T}_{0}$ resp. $\textup{deg}(w)<\alpha$. Assume that $\llbracket \phi \rrbracket^J_h=w$ resp. $\llbracket \phi \rrbracket^I_h=w$ resp. $\llbracket \phi \rrbracket^{K_1}_h=w$ (where $K_1=I$ and $K_2=J$ or  $K_1=J$ and $K_2=I$). Using Definition \ref{wfunktion} we get that $\mathcal{T}_{\alpha-1}\leq \llbracket A\rrbracket_h^J\leq \mathcal{T}_{0}$ resp. $\mathcal{F}_{0}\leq \llbracket A\rrbracket_h^I\leq \mathcal{F}_{\alpha-1}$ resp. $\textup{deg}(\llbracket A\rrbracket_h^{K_{1}})<\alpha-1$. Then, from the third part of $I_H(A),$ $\llbracket A\rrbracket_h^J=\llbracket A\rrbracket_h^I$ resp. $\llbracket A\rrbracket_h^I=\llbracket A\rrbracket_h^J$ resp. $\llbracket A\rrbracket_h^{K_1}=\llbracket A\rrbracket_h^{K_2}$. Finally, using Definition \ref{wfunktion}, we get that $\llbracket \phi \rrbracket^I_h=w$ resp. $\llbracket \phi \rrbracket^J_h=w$ resp. $\llbracket \phi \rrbracket^{K_2}_h=w$.
\\\\Before we can go on with the next case, we must prove the following technical lemma. 
\begin{lemma}\label{Hilfssatz zu Fortsetzungssatz1}We use the same assumptions as in Theorem \ref{Fortsetzungssatz1}. Let $\mathcal{I}$  be a set of indices, $A_i$ ($i\in \mathcal{I}$) a formula provided that $I_H(A_i)$ and $h_i$ ($i\in \mathcal{I}$) an assignment. We define $\textup{inf}_K:=\textup{inf}\{\llbracket A_i\rrbracket_{h_i}^K;i\in \mathcal{I}\}$ and $\textup{sup}_K:=\textup{sup}\{\llbracket A_i\rrbracket_{h_i}^K;i\in \mathcal{I}\}$ (where $K=I,J$). Then the following holds:\begin{enumerate}\item $\textup{inf}_J=\mathcal{F}_\gamma\;\;\Rightarrow \;\textup{inf}_I=\mathcal{F}_\gamma$ (for all $\gamma\leq\alpha$)
\item $\textup{inf}_I=\mathcal{T}_\gamma\;\;\;\Rightarrow\;\textup{inf}_J=\mathcal{T}_\gamma$ (for all $\gamma\leq\alpha$)
\item $\textup{inf}_I=w\;\;\;\;\Leftrightarrow\,\,\textup{inf}_J=w$ (for all $w$ provided that $\textup{deg}(w)<\alpha$) \item $\textup{sup}_J=\mathcal{F}_\gamma\;\Rightarrow \,\textup{sup}_I=\mathcal{F}_\gamma$ (for all $\gamma\leq\alpha$)
\item $\textup{sup}_{I}=\mathcal{T}_\gamma\;\;\Rightarrow\,\textup{sup}_J=\mathcal{T}_\gamma$ (for all $\gamma\leq\alpha$)
\item $\textup{sup}_I=w\;\;\;\Leftrightarrow\,\textup{sup}_J=w$ (for all $w$ provided that $\textup{deg}(w)<\alpha$)
\end{enumerate}
\end{lemma} 
\begin{proof} \textit{1.:} Assume that $\textup{inf}_J=\mathcal{F}_\gamma$. Using Lemma \ref{sup inf existenz annahme} we get that there exists an $i_0$  such that $\llbracket A_{i_0}\rrbracket_{h_{i_{0}}}^J=\mathcal{F}_\gamma$. Then, from the first part of $I_H(A_{i_{0}}),$ $\llbracket A_{i_0}\rrbracket_{h_{i_{0}}}^I=\mathcal{F}_\gamma$. This implies that $\llbracket A_{i_0}\rrbracket_{h_{i_{0}}}^I\leq\llbracket A_{i}\rrbracket_{h_{i}}^I$ for all $i\in \mathcal{I}$. (Since otherwise we had that there exists a $j_0\in \mathcal{I}$ such that $\llbracket A_{j_0}\rrbracket_{h_{j_{0}}}^I<\mathcal{F}_\gamma$. Then, using the third part of  $I_H(A_{j_0})$, it would also be $\llbracket A_{j_0}\rrbracket_{h_{j_{0}}}^J<\mathcal{F}_\gamma$. But this contradicts our assumption $\textup{inf}_J=\mathcal{F}_\gamma$.) Finally, we get that $\textup{inf}_I=\mathcal{F}_\gamma$.
\\\textit{2.:} Assume now, that $\textup{inf}_I=\mathcal{T}_\gamma$. Then $\mathcal{T}_\gamma\leq\llbracket A_{i}\rrbracket_{h_{i}}^I$ for all $i\in \mathcal{I}$. Using part two of $I_H(A_i)$, we get that $\llbracket A_{i}\rrbracket_{h_{i}}^I=\llbracket A_{i}\rrbracket_{h_{i}}^J$ for all $i$. This implies $\textup{inf}_J=\mathcal{T}_\gamma$.
\\\textit{3.:} Due to 1. and 2., it only remains to show $(\textup{inf}_J=\mathcal{T}_\gamma\;\Rightarrow \;\textup{inf}_I=\mathcal{T}_\gamma)$ and $(\textup{inf}_I=\mathcal{F}_\gamma\;\Rightarrow \;\textup{inf}_J=\mathcal{F}_\gamma)$ for $\gamma<\alpha$. Assume that $\textup{inf}_J=\mathcal{T}_\gamma$ (where $\gamma<\alpha$). Then $\mathcal{T}_\gamma\leq\llbracket A_{i}\rrbracket_{h_{i}}^J$ for all $i\in \mathcal{I}$ and this implies, using the third part of $I_H(A_i)$, $\llbracket A_{i}\rrbracket_{h_{i}}^J=\llbracket A_{i}\rrbracket_{h_{i}}^I$ for all $i$. Finally, we get that $\textup{inf}_I=\mathcal{T}_\gamma$. 
\\ For the latter case assume that $\textup{inf}_I=\mathcal{F}_\gamma$  ($\gamma<\alpha$). Then there exists an $i_0$ such that $\llbracket A_{i_0}\rrbracket_{h_{i_{0}}}^I=\mathcal{F}_\gamma$ (Lemma \ref{sup inf existenz annahme}). Then, using the third part of $I_H(A_{i_{0}})$, we get that $\llbracket A_{i_0}\rrbracket_{h_{i_{0}}}^J=\mathcal{F}_\gamma$. This implies that $\llbracket A_{i_0}\rrbracket_{h_{i_{0}}}^J\leq\llbracket A_{i}\rrbracket_{h_{i}}^J$ for all $i\in \mathcal{I}$. (Since otherwise we had that there exists a $j_0\in \mathcal{I}$ such that $\llbracket A_{j_0}\rrbracket_{h_{j_{0}}}^I<\mathcal{F}_\gamma$, see proof of statement 1.) Finally, we get that $\textup{inf}_J=\mathcal{F}_\gamma$.\\
We will not give the\ proofs of 4., 5., and 6. here, because they are similar to 1., 2., and 3..\qed
\end{proof}
\textit{Case 4:} $\phi=A\wedge B$. Assume that $I_H(A)$ and $I_H(B)$. Let $h$ be an arbitrary assumption. We define $\mathcal{I}:=\{1,2\}$, $h_1:=h$, $h_2:=h$, $A_1:=A$ and $A_2:=B$. Then $I_H(A_i)$ for $i=1,2$, $\llbracket \phi \rrbracket^J_h=\textup{min}\{\llbracket A\rrbracket^J_h,\llbracket B\rrbracket^J_h\}=\textup{inf}_{J}$ and $\llbracket \phi \rrbracket^I_h=\textup{min}\{\llbracket A\rrbracket^I_h,\llbracket B\rrbracket^I_h\}=\textup{inf}_{I}$. Then, using 1., 2. and 3. of Lemma \ref{Hilfssatz zu Fortsetzungssatz1}, we get that 1., 2. and 3. of Theorem \ref{Fortsetzungssatz1} hold.
\\\textit{Case 5:} $\phi=A\vee B$. Replace $\textup{min}$ by $\textup{max}$ and $\textup{inf}$ by $\textup{sup}$ in the proof above and use 4., 5. and 6. of Lemma \ref{Hilfssatz zu Fortsetzungssatz1} instead of 1., 2. and 3..
\\\textit{Case 6:} $\phi=\forall v(A)$. Assume that $I_H(A)$ and let $h$ be an arbitrary assumption.
\\ We define $\mathcal{I}:=\{u;\;u\in H_U\}$, $h_u:=h[v\mapsto u]$ and $A_u:=A$ for all $u\in H_U$. Then $I_H(A_u)$ for all $u\in \mathcal{I}$, $\llbracket \phi \rrbracket^J_h=\textup{inf}\{\llbracket A\rrbracket^J_{h[v\mapsto u]};\;u\in H_U\}=\textup{inf}_{J}$, and $\llbracket \phi \rrbracket^I_h=\textup{inf}\{\llbracket A\rrbracket^I_{h[v\mapsto u]};\;u\in H_U\}=\textup{inf}_{I}$. Then, using 1., 2. and 3. of Lemma \ref{Hilfssatz zu Fortsetzungssatz1}, we get that 1., 2. and 3. of Theorem \ref{Fortsetzungssatz1} hold.
\\\textit{Case 7:} $\phi=\exists v(A)$. Replace $\textup{inf}$ by $\textup{sup}$ in the proof above and use 4., 5. and 6. of Lemma \ref{Hilfssatz zu Fortsetzungssatz1} instead of 1., 2. and 3..
\qed\end{proof}
\begin{lemma} \label{alphamonotonie} The immediate consequence operator $T_P$ of a given program $P$ is $\alpha$-monotonic for all countable ordinals $\alpha$.
\end{lemma}
\begin{proof} The proof is by transfinite induction on $\alpha$. Assume the lemma holds for all $\beta<\alpha$ (induction hypothesis). We demonstrate that it also holds for $\alpha$.
Let $I,J$ be two interpretations such that $I\sqsubseteq_\alpha J$. Then, using the induction hypothesis, we get that \begin{equation}T_P(I)=_\beta T_P(J)\text{ for all }\beta<\alpha \label{lemmamono(S1)}.\end{equation} 
It remains to show that  $T_P(I)\parallel \mathcal{T}_\alpha\subseteq T_P(J)\parallel \mathcal{T}_\alpha$ and that  $T_P(J)\parallel \mathcal{F}_\alpha\subseteq T_P(I)\parallel \mathcal{F}_\alpha$. For the first statement assume that $T_P(I)(A)=\mathcal{T}_\alpha$ for some $A\in H_B.$  Then, using Lemma \ref{sup inf existenz annahme}, there exists a ground instance $A\leftarrow \phi$ of $P$ such that $\llbracket \phi\rrbracket ^I=\mathcal{T}_{\alpha}$. But then, by      Theorem \ref{Fortsetzungssatz1}, $\llbracket \phi\rrbracket ^J=\mathcal{T}_{\alpha}$. This implies $\mathcal{T}_\alpha\leq T_P(J)(A)$. But this implies $\mathcal{T}_\alpha= T_P(J)(A)$. (Since $\mathcal{T}_\alpha< T_P(J)(A)$, using  (\ref{lemmamono(S1)}), would imply $\mathcal{T}_\alpha< T_P(I)(A)$.) For the latter statement assume that $T_P(J)(A)=\mathcal{F}_\alpha$ for some $A\in H_B$. This implies that $\llbracket \phi\rrbracket^J\leq \mathcal{F}_\alpha$ for every ground instance $A\leftarrow \phi$ of $P$. But then, using again Theorem \ref{Fortsetzungssatz1}, we get that   $\llbracket \phi\rrbracket^I=\llbracket \phi\rrbracket^J$ for every ground instance $A\rightarrow \phi$ of $P$. Finally, this implies also $T_P(I)(A)=\mathcal{F}_\alpha$.\qed\end{proof} 
\begin{remark} The immediate consequence operator $T_P$ is not monotonic with respect to $\sqsubseteq_\infty$. Consider the program $P=\{A\leftarrow\neg A\}$ and the interpretations $I_1$ and $I_2$ given by $I_1:=\{A\mapsto \mathcal{F}_0\}$ and $I_2:=\{A\mapsto 0\}$. Obviously, $I_1\sqsubset_0 I_2$ and hence $I_1\sqsubseteq_\infty I_2$. Using Definition \ref{def immediate consequence}, we get that $T_P(I_1)=sup\{\llbracket\neg A\rrbracket^{I_1}\}=\mathcal{T}_1$ and $T_P(I_2)=sup\{\llbracket \neg A\rrbracket^{I_2}\}=0$. This implies $T_P(I_2)\sqsubset_1 T_P(I_1)$ (i.e., $T_P(I_1)\sqsubseteq_\infty T_P(I_2)$ does not hold). 
\end{remark}
\section{Construction of the Minimum Model}
In this section we show how to construct the interpretation $M_P$     of a given formula-based logic program $P$. We will give the proof that $M_P$ is a model of $P$ and that it is the least of all models of $P$ in the next section. In \cite{RondogiannisWadge} the authors give a clear informal description of the following construction: \\``As a first approximation to $M_P$, we start (...) iterating the $T_P$ on $\emptyset$ until both the set of atoms that have a $\mathcal{F}_0$ value and the set of atoms having $\mathcal{T}_0$ value, stabilize. We keep all these atoms whose values have stabilized and reset the values of all remaining atoms to the next false value (namely $\mathcal{F}_1$). The procedure is repeated until the $\mathcal{F}_1$ and $\mathcal{T}_1$ values stabilize, and we reset the remaining atoms to a value equal to $\mathcal{F}_2$, and so on. Since the Herbrand Base of $P$ is countable, there exists a countable ordinal $\delta$ for which this process will not produce any new atoms having $\mathcal{F}_\delta$ or $\mathcal{T}_\delta$ values. At this point we stop iteration and reset all remaining atoms to the value 0." 
\begin{definition}\label{limes aleph_1}\textup{Let $P$ be a program, $I$ an interpretation, and $\alpha\in\aleph_1$ such that $I\sqsubseteq_\alpha T_P(I)$. We define by recursion on the ordinal $\beta\in\Omega$ the interpretation $T^\beta_{P,\alpha}(I)$ as follows: \\$T_{P,\alpha}^0(I):=I$ and if  $\beta$ is a successor ordinal, then $T^{\beta}_{P,\alpha}:=T_P(T^{\beta-1}_{P,\alpha})$. If $0<\beta$ is a limit ordinal and $A\in H_B$, then
$$T^{\beta}_{P,\alpha}(I)(A):=\begin{cases} I(A), &\text{if }\textup{deg}(I(A))<\alpha\\
 \mathcal{T}_{\alpha}, &\text{if }A\in\bigcup_{\gamma\in\beta} T^\gamma_{P,\alpha}(I)\|\mathcal{T}_{\alpha}\\ \mathcal{F}_{\alpha}, &\text{if }A\in\bigcap_{\gamma\in\beta} T^\gamma_{P,\alpha}(I)\|\mathcal{F}_{\alpha}\\
\mathcal{F}_{\alpha+1},&\text{otherwise}\end{cases}.$$
}\end{definition}
\begin{lemma}\label{Ketteneigenschaften} Let $P$ be a program, $I$ an interpretation and $\alpha\in\aleph_1$ such that $I\sqsubseteq_\alpha T_P(I)$. Then the following holds:\begin{enumerate}
\item For all limit ordinals $0<\gamma\in\Omega$ and all interpretations $M$ the condition $\forall\beta<\gamma:T_{P,\alpha}^\beta(I)\sqsubseteq_\alpha M$ implies $T^\gamma_{P,\alpha}(I)\sqsubseteq_{\alpha} M$.
\item For all $\beta\leq\gamma\in\Omega$ the property $T^\beta_{P,\alpha}(I)\sqsubseteq_{\alpha}T^{\gamma}_{P,\alpha}(I)$ holds.
\end{enumerate}
\end{lemma}
\begin{proof} \textbf{1.}: The proof follows directly from the above definition.
\\\textbf{2.}: One can prove the second statement with induction, using the assumption $I\sqsubseteq_\alpha T_P(I)$, the fact that $T_P$ is $\alpha$-monotonic, the fact that $\sqsubseteq_\alpha$ is transitive and at limit stage the first statement of this lemma.\qed 
\end{proof}
At this point, we have to consider a theorem of Zermelo-Fraenkel axiomatic set theory with the Axiom of Choice (ZFC). In the case of normal logic programs this theorem is not necessary, because in the bodies of normal logic programs do not appear ``$\forall$" or ``$\exists$". One can find the proof of the theorem in \cite{Jech}.
\begin{definition}\textup{Let $\alpha>0$ be a limit ordinal. We say that an increasing $\beta$-sequence $(\alpha_\zeta)_{\zeta<\beta}$, $\beta$ limit ordinal, is \textit{cofinal} in $\alpha$ if $\textup{sup}\{\alpha_\zeta;\zeta<\beta\}=\alpha$. Similarly, $A\subseteq\alpha$ is \textit{cofinal} in $\alpha$ if $\textup{sup}A=\alpha$. If $\alpha$ is an infinite limit ordinal, the \textit{cofinality} of $\alpha$ is $cf(\alpha)=$ ``the least limit ordinal $\beta$ such that there is an increasing $\beta$-sequence $(\alpha_\zeta)_{\zeta<\beta}$ with $\textup{sup}\{\alpha_\zeta;\zeta<\beta\}=\alpha$''. An infinite cardinal $\aleph_\alpha$ is \textit{regular} if $cf(\aleph_\alpha)=\aleph_\alpha$.}\end{definition}
\begin{theorem}\label{aleph1regulär} Every cardinal of the form $\aleph_{\alpha+1}$ is regular. Particularly, $\aleph_1$ is regular.\end{theorem}
\begin{theorem}[Extension Theorem II]\label{Fortsetzungssatz2}Let $P$ be a program, $I$ an interpretation, and $\alpha\in\aleph_1$ such that $I\sqsubseteq_\alpha T_P(I)$. Then for every formula $\phi\in \textup{Form}$ and every assignment $h$ the following hold: \begin{enumerate}
\item $\llbracket \phi\rrbracket_h^{T^{\aleph_1}_{P,\alpha}(I)}=\llbracket \phi\rrbracket_h^I, \;\;\quad\quad\quad\quad\;\;\text{if $\textup{deg}(\llbracket \phi\rrbracket_h^I)<\alpha$ \,\,\quad\quad\quad\quad\quad\quad\quad\;\,\textup{(C1)}}$

\item $\llbracket \phi\rrbracket_h^{T^{\aleph_1}_{P,\alpha}(I)}=\mathcal{T}_\alpha, 
\text{ \quad\quad\quad\;\;\quad\quad if $\llbracket\phi\rrbracket_h^{T^{i}_{P,\alpha}(I)}=\mathcal{T}_\alpha$ for some $i\in\aleph_1$ \,\,\textup{(C2)}}$

\item $\llbracket \phi\rrbracket_h^{T^{\aleph_1}_{P,\alpha}(I)}=\mathcal{F}_\alpha,\text{ \quad\quad\quad\quad\;\quad if $\llbracket \phi\rrbracket_h^{T^{i}_{P,\alpha}(I)}=\mathcal{F}_\alpha$ for all $i\in\aleph_1$ \,\,\;\;\;\,\textup{(C3)}}$

\item $\mathcal{F}_\alpha<\llbracket \phi\rrbracket_h^{T^{\aleph_1}_{P,\alpha}(I)}<\mathcal{T}_\alpha\quad\quad\quad\,\,\Leftrightarrow\text{not\textup{(C1)} and not\textup{(C2)} and not\textup{(C3)}}$
\end{enumerate}
\end{theorem}
\begin{proof} \textit{1. and 2.:} We get this using Lemma \ref{Ketteneigenschaften} and Theorem \ref{Fortsetzungssatz1}. \\\textit{3.:} We show this by induction on $\phi$. We define $I_i:=T^{i}_{P,\alpha}(I)$ and $I_\infty:=T^{\aleph_1}_{P,\alpha}(I)$. Moreover, we use $I_H(X)$ as an abbreviation for \begin{center}``for all assignments $g$ the property $\forall i\in\aleph_1(\llbracket X\rrbracket_g^{I_i}=\mathcal{F}_\alpha )\Rightarrow\llbracket X\rrbracket_g^{I_\infty}=\mathcal{F}_\alpha$ holds''.\end{center}
\textit{Case 1:} $\phi=P_k(t_1,...,t_{s_k})$ or $=\top,\bot$. This follows directly from Definition \ref{limes aleph_1} respectively from Definition \ref{wfunktion}.
\\\textit{Case 2:} $\phi=\neg(A)$. Assuming $\forall i\in\aleph_1$: $\llbracket \phi\rrbracket_{h}^{I_i}=\mathcal{F}_{\alpha}$ we conclude $\forall i\in\aleph_1:\llbracket A\rrbracket_h^{I_i}=\mathcal{T}_{\alpha-1}$. Then, by Theorem \ref{Fortsetzungssatz1}, we get $\llbracket A\rrbracket_h^{I_\infty}=\mathcal{T}_{\alpha-1}$ and this implies $\llbracket \phi\rrbracket_h^{I_\infty}=\mathcal{F}_{\alpha}$.
\\\textit{Case 3:}  $\phi=A\wedge B$ or $\phi=A\vee B$. The following cases         are more general than this case. Therefore, we will not give a proof here.\\\textit{Case 4:}  $\phi=\exists v(A)$. We assume that $I_H(A)$ and for every $i\in\aleph_1$ we assume that $\llbracket \phi\rrbracket_h^{I_i}=\mathcal{F}_{\alpha}$. This implies $\textup{sup}\{\llbracket A\rrbracket_{h[v\mapsto u]}^{I};\;u\in H_U\}=\mathcal{F}_{\alpha}$ as well as $\forall i\in\aleph_1\forall u\in H_U:\llbracket A\rrbracket_{h[v\mapsto u]}^{I_i}\leq \mathcal{F}_\alpha$. Now we show by case distinction that $\forall u\in H_U:\llbracket A\rrbracket_{h[v\mapsto u]}^{I_\infty}=\llbracket A\rrbracket_{h[v\mapsto u]}^I$ and this obviously implies $\llbracket \phi\rrbracket_h^{I_\infty}=\mathcal{F}_{\alpha}$. First we consider the case  $\llbracket A\rrbracket_{h[v\mapsto u]}^I<\mathcal{F}_\alpha$. Then, using Lemma \ref{Ketteneigenschaften} and Theorem \ref{Fortsetzungssatz1}, we get that $\llbracket A\rrbracket_{h[v\mapsto u]}^{I_\infty}=\llbracket A\rrbracket_{h[v\mapsto u]}^I.$ At least, we consider the other case $\llbracket A\rrbracket_{h[v\mapsto u]}^I=\mathcal{F}_\alpha$. We know that $\forall i\in\aleph_1:\llbracket A\rrbracket_{h[v\mapsto u]}^{I_i}\leq \mathcal{F}_\alpha$. But this implies $\forall i\in\aleph_1:\llbracket A\rrbracket_{h[v\mapsto u]}^{I_i}= \mathcal{F}_\alpha$, since $\exists i\in\aleph_1:\llbracket A\rrbracket_{h[v\mapsto u]}^{I_i}< \mathcal{F}_\alpha$ would imply (using Lemma \ref{Ketteneigenschaften} and Theorem \ref{Fortsetzungssatz1}) $\llbracket A\rrbracket_{h[v\mapsto u]}^I<\mathcal{F}_\alpha$, which is a contradiction. Finally, we get, by $I_H(A)$, that $\llbracket A\rrbracket_h^{I_\infty}=\mathcal{F}_{\alpha}=\llbracket A\rrbracket_{h[v\mapsto u]}^I$.\\\textit{Case 5:} $\phi=\forall v(A)$. We assume that $I_H(A)$ and for every $i\in\aleph_1$ we assume that $\llbracket \phi\rrbracket_h^{I_i}=\mathcal{F}_{\alpha}$. Then $\forall i\in\aleph_1:\textup{inf}\{\llbracket A\rrbracket_{h[v\mapsto u]}^{I_i};\;u\in H_U\}=\mathcal{F}_{\alpha}.$ This implies, using Lemma \ref{sup inf existenz annahme}, $\forall i\in\aleph_1\exists u\in H_U:\llbracket A\rrbracket_{h[v\mapsto u]}^{I_i}=\mathcal{F}_\alpha$. Next we choose for every $i\in \aleph_1$ an atom $u_i\in H_U$ with $\llbracket A\rrbracket_{h[v\mapsto u_i]}^{I_i}=\mathcal{F}_\alpha$ (Remark: We do not need the Axiom of Choice because $H_U$ is countable). Then, using Lemma \ref{Ketteneigenschaften} and Theorem \ref{Fortsetzungssatz1}, $\forall i\in\aleph_1\forall j\leq i\in\aleph_{1}:\llbracket A\rrbracket_{h[v\mapsto u_i]}^{I_j}=\mathcal{F}_\alpha$ . This implies that the mapping $$\zeta:\left\{ u_i;i\in\aleph_1 \right\}\rightarrow\aleph_1\cup\{\aleph_1\}:u\mapsto \begin{cases} \textup{min}\{j\in\aleph_1; \llbracket A\rrbracket_{h[v\mapsto u]}^{I_j}\neq \mathcal{F}_\alpha\}, &\text{if $\textup{min}$ exists}\\
 \aleph_1, &\text{otherwise} \end{cases}$$
 has the properties $\forall i\in\aleph_1:\zeta(u_i)>i$ and $\textup{sup}\{\zeta (u_i);i\in\aleph_1\}=\aleph_1$. We assume now that  $\forall u\in H_U\exists j\in\aleph_1:\llbracket A\rrbracket_{h[v\mapsto u]}^{I_j}\neq \mathcal{F}_\alpha$. Then $\zeta(\{u_i;i\in\aleph_1\})$ is a countable subset of $\aleph_1$ and moreover cofinal in $\aleph_1$. But this is a contradiction to Theorem \ref{aleph1regulär}. Therefore we know that there exists an atom $u^*\in H_U$  such that $\forall i\in\aleph_1:\llbracket A\rrbracket_{h[v\mapsto u^*]}^{I_i}=\mathcal{F}_\alpha.$  Then, using  $I_H(A)$, we get that $\llbracket A\rrbracket_{h[v\mapsto u^*]}^{I_\infty}=\mathcal{F}_\alpha.$ This implies $\llbracket \phi\rrbracket_h^{I_\infty}\leq \mathcal{F}_\alpha$ and finally, using $\llbracket \phi\rrbracket_h^{I}=\mathcal{F}_{\alpha}$, Lemma \ref{Ketteneigenschaften} and Theorem \ref{Fortsetzungssatz1}, we get that $\llbracket \phi\rrbracket_h^{I_\infty}= \mathcal{F}_\alpha.$
\\\\\textit{ 4.:}``$\Rightarrow$'': We prove this by the method of contrapositive. We assume that (C1) or (C2) or (C3). Then, using 1., 2., and 3., we get that not($\mathcal{F}_\alpha<\llbracket \phi\rrbracket_h^{I_\infty}<\mathcal{T}_\alpha$) holds. 
\\``$\Leftarrow$'': We shall first consider the following Lemma.
\begin{lemma} \label{Hilf FS II}Under the same conditions as in Theorem \ref{Fortsetzungssatz2}  for every formula $\phi\in \textup{Form}$ and every assignment $h$ the following hold: $$\llbracket \phi\rrbracket_h^{I_\infty}=\mathcal{T}_\alpha 
\text{\quad$\Rightarrow\quad\llbracket \phi\rrbracket_h^{I_i}=\mathcal{T}_\alpha$ for some $i\in\aleph_1$}$$
\end{lemma}
\begin{proof} This proof is similar to the proof of Theorem \ref{Fortsetzungssatz2} statement 3. (see Appendix).\qed\end{proof}
We prove ``$\Leftarrow$'' also by the method of contrapositive. We assume that $\mathcal{F}_\alpha<\llbracket \phi\rrbracket^{I_\infty}_{h}<\mathcal{T}_\alpha$  does not hold. We consider the three possible cases $\textup{deg}(\llbracket \phi\rrbracket^{I_\infty}_{h})<\alpha$, $\llbracket \phi\rrbracket^{I_\infty}_{h}=\mathcal{F}_\alpha$, and $\llbracket \phi\rrbracket^{I_\infty}_{h}=\mathcal{T}_\alpha$. Let us consider the first case (resp. the second case). Then, using Lemma \ref{Ketteneigenschaften} and Theorem \ref{Fortsetzungssatz1}, (C1) (resp. (C3)) holds. Now, we  consider  the latter case. Using Lemma \ref{Hilf FS II} we get that (C2) holds. Finally, in every case (C1) or (C2) or (C3) holds.\qed\end{proof}
\begin{definition}\label{Def vereinigung}\textup{Let $\alpha$ be a countable ordinal and for every $\gamma<\alpha$  let $I_\gamma$ be an interpretation such that $\forall\zeta\leq\gamma: I_\zeta=_{\zeta}I_\gamma.$ Then the union of the interpretations $I_\gamma$ ($\gamma<\alpha$) is a well-defined interpretation and  given by the following definition:
$$\bigsqcup_{\gamma<\alpha} I_\gamma\,\,(A):=\begin{cases}\mathcal{F}_\zeta, & \text{if }\zeta<\alpha\;\&\;I_\zeta(A)=\mathcal{F}_\zeta \\
\mathcal{T}_\zeta, & \text{if }\zeta<\alpha\;\&\;I_\zeta(A)=\mathcal{T}_\zeta \\
\mathcal{F}_\alpha, & \text{otherwise} \\
\end{cases} \;\;\;\;\;\;\;(A\in H_B)$$
}\end{definition}
\begin{remark} Using $\forall\zeta\leq\gamma: I_\zeta=_{\zeta}I_\gamma$  it is easy to prove that the union  $\bigsqcup_{\gamma<\alpha} I_\gamma$ is a well-defined interpretation. Particularly if $\alpha=0,$  then the union is equal to the interpretation that     maps all atoms of $H_B$ to the truth value $\mathcal{F}_0$. This interpretation is sometimes denoted by $\emptyset$.\end{remark} 
\begin{lemma}\label{KonstrHilf1}Let $P$ be a program, $\alpha$  be a countable ordinal and for all $\gamma<\alpha$ an interpretation $I_\gamma$ is given such that $\forall\zeta<\gamma: I_\zeta=_{\zeta}I_\gamma.$ Then the following holds: $$\forall\gamma<\alpha\left( I_\gamma\sqsubseteq_{\gamma+1}T_P(I_\gamma)\right)\;\Rightarrow\;\bigsqcup_{\gamma<\alpha} I_\gamma\sqsubseteq_{\alpha}T_P(\bigsqcup_{\gamma<\alpha} I_\gamma)$$ 
\end{lemma}
\begin{proof}
 We assume that \begin{equation}\forall\gamma<\alpha\left( I_\gamma\sqsubseteq_{\gamma+1}T_P(I_\gamma)\right) \label{lemma6star}.\end{equation} First, we prove that $\forall\beta<\alpha:\;\bigsqcup_{\gamma<\alpha} I_\gamma=_{\beta}T_P(\bigsqcup_{\gamma<\alpha} I_\gamma).$ For all $\beta<\alpha$ we know that $I_\beta=_\beta\bigsqcup_{\gamma<\alpha} I_\gamma$. Then, using Lemma \ref{alphamonotonie}, $\forall\beta<\alpha:T_P({I_\beta})=_{\beta}T_P(\bigsqcup_{\gamma<\alpha} I_\gamma)$. This implies for all $\beta<\alpha$ the property $\bigsqcup_{\gamma<\alpha} I_\gamma=_\beta I_\beta=_\beta^{(\ref{lemma6star})} T_P({I_\beta})=_{\beta}T_P(\bigsqcup_{\gamma<\alpha} I_\gamma).$ We  know that $\bigsqcup_{\gamma<\alpha} I_\gamma$ does not map to truth values $w$ such that $\mathcal{F}_\alpha<w\leq \mathcal{T}_\alpha$. And this obviously implies $\;\bigsqcup_{\gamma<\alpha} I_\gamma\sqsubseteq_{\alpha}T_P(\bigsqcup_{\gamma<\alpha} I_\gamma).$\qed\end{proof}
\begin{lemma}\label{KonstruktionsHilfe2}Let $P$ be a program, $\alpha$ a countable ordinal, and $I$ an interpretation. Then the following holds: 
$$I\sqsubseteq_\alpha T_P(I)\;\;\Rightarrow\;\;T^{\aleph_1}_{P,\alpha}(I)\sqsubseteq_{\alpha+1}T_P(T^{\aleph_1}_{P,\alpha}(I))$$
\end{lemma}
\begin{proof} Again, we define $I_i:=T^{i}_{P,\alpha}(I)$, $I_\infty:=T^{\aleph_1}_{P,\alpha}(I)$. Let us assume that $I\sqsubseteq_\alpha T_P(I)$. First we prove $I_\infty\sqsubseteq_{\alpha}T_P(I_\infty).$ Using Lemma \ref{Ketteneigenschaften} we get that $\forall\gamma<\aleph_1:I_\gamma\sqsubseteq_\alpha I_\infty$. Then, using Lemma \ref{alphamonotonie}, $\gamma\in\aleph_1:I_{\gamma+1}\sqsubseteq_\alpha T_P(I_\infty)$. Using again Lemma \ref{Ketteneigenschaften} and the transitivity of $\sqsubseteq_{\alpha}$   we get that $\forall\gamma<\aleph_1: I_{\gamma}\sqsubseteq_\alpha T_P(I_\infty)$.
Then, using the first part of Lemma \ref{Ketteneigenschaften}, $I_\infty\sqsubseteq_{\alpha}T_P(I_\infty).$
\\Let us prove now $T_P(I_\infty)\sqsubseteq_{\alpha}I_\infty$. It remains to show \begin{equation}I_\infty\parallel \mathcal{F}_\alpha\subseteq T_P(I_\infty)\parallel \mathcal{F}_{\alpha}\label{Lemma7Case1}\end{equation} as well as \begin{equation}T_P(I_\infty)\Vert \mathcal{T}_\alpha\subseteq I_\infty\Vert \mathcal{T}_\alpha \label{Lemma7(Case2)}.\end{equation} Firstly, let us prove that (\ref{Lemma7Case1}) holds and therefore we assume that $I_\infty(A)=\mathcal{F}_\alpha$ for some $A\in H_B$. Then, using the definition of  $I_\infty$, we get that for all $i\in\aleph_1$ the following holds: \begin{equation}I_{i}(A)=\mathcal{F}_\alpha\;\label{Lemma7(S1)}\end{equation} Let $A\leftarrow \phi$ be an arbitrary ground instance of $P$. We prove now that the property $\llbracket \phi\rrbracket^{I_\infty}=\llbracket \phi\rrbracket^{I}$ holds. Then, using  (\ref{Lemma7(S1)}) and the definition of the immediate consequence operator $T_P$, we get that $\mathcal{F}_\alpha=I_{1}(A)=\textup{sup}\{\llbracket C\rrbracket^{I};\;A\leftarrow C\in P_G\}$. This implies either $\llbracket \phi\rrbracket^{I}<\mathcal{F}_\alpha$ or $\llbracket \phi\rrbracket^{I}=\mathcal{F}_\alpha$. We consider the first case. Then, using Theorem \ref{Fortsetzungssatz2}, we get that $\llbracket \phi\rrbracket^{I_\infty}=\llbracket \phi\rrbracket^{I}$. In the latter case, using again (\ref{Lemma7(S1)}), we get that for all $i\in\aleph_1$ the property $\mathcal{F}_\alpha=I_{i+1}(A)=\textup{sup}\{\llbracket C\rrbracket^{I_i};\;A\leftarrow C\in P_G\}$ holds. This obviously implies $\forall i\in\aleph_1:\llbracket \phi\rrbracket^{I_{i}}\leq \mathcal{F}_\alpha$. Then, using Lemma \ref{Ketteneigenschaften} and Theorem \ref{Fortsetzungssatz1}, we get that $\forall i\in\aleph_1:\llbracket \phi\rrbracket^{I_{i}}= \mathcal{F}_\alpha$. But then the third part of Theorem \ref{Fortsetzungssatz2} finally implies that $\llbracket \phi\rrbracket^{I_\infty}= \mathcal{F}_\alpha=\llbracket \phi\rrbracket^{I}.$ \\Thus the above argumentation implies that for all  $A\leftarrow \phi$ in $P_{G}$ the equation $\llbracket \phi\rrbracket^{I_\infty}=\llbracket \phi\rrbracket^{I}$ holds.
This implies $\mathcal{F}_\alpha=I_{1}(A)=\textup{sup}\{\llbracket \phi\rrbracket^{I};\;A\leftarrow \phi\in P_G\}=\textup{sup}\{\llbracket \phi\rrbracket^{I_\infty};\;A\leftarrow \phi\in P_G\}=T_P(I_\infty)(A).$
\\Secondly, let us prove (\ref{Lemma7(Case2)}) and therefore we assume now that $T_{P}(I_\infty)(A)=\mathcal{T}_\alpha$ for some $A\in H_B$. Then $\textup{sup}\{\llbracket \phi\rrbracket^{I_\infty};\;A\leftarrow \phi\in P_G\}=\mathcal{T}_\alpha.$ This and Lemma \ref{sup inf existenz annahme} allow us to choose a ground instance $A\leftarrow \phi$ such that $\llbracket \phi\rrbracket^{I_\infty}=\mathcal{T}_\alpha$. Then, using Lemma \ref{Hilf FS II}, we can choose an ordinal $i_0\in\aleph_1$ such that $\llbracket \phi\rrbracket^{I_{i_0}}=\mathcal{T}_\alpha$. This implies $\llbracket A\rrbracket^{I_{i_0+1}}\geq \mathcal{T}_\alpha$. We know $I_{i_0+1}\sqsubseteq_{\alpha}T_P(I_\infty)$ by Lemma \ref{Ketteneigenschaften} and Lemma \ref{alphamonotonie}. But then, using Theorem 1 and the assumption of this case, $\llbracket A\rrbracket^{I_{i_0+1}}= \mathcal{T}_\alpha$ must hold. Finally, using the second part of Theorem \ref{Fortsetzungssatz2}, we get that $I_\infty(A)=\mathcal{T}_\alpha.$ \\The argumentation above  implies that  $I_\infty=_{\alpha}T_P(I_\infty)$.
We  know that $I_\infty$ does not map to truth values $w$ such that $\mathcal{F}_{\alpha+1}<w\leq \mathcal{T}_{\alpha+1}$. And this obviously implies $I_\infty\sqsubseteq_{\alpha+1}T_P(I_\infty)$.\qed\end{proof}
\begin{definition}\textup{\label{defapprox}Let $P$ be a program. We define by recursion on the  countable  ordinal $\alpha$ the approximant $M_\alpha$ of $P$ as follows:
$$M_\alpha:=\begin{cases}T^{\aleph_1}_{P,\alpha}(\bigsqcup_{\gamma<\alpha} M_{\gamma}),&\text{if $\begin{array}{c}
\forall\gamma<\alpha\forall\zeta<\gamma\left(M_\zeta=_{\zeta}M_\gamma\right) \;\& \\
\;\bigsqcup_{\gamma<\alpha} M_\gamma\sqsubseteq_{\alpha}T_P(\bigsqcup_{\gamma<\alpha} M_\gamma) \\
\end{array}$}\\\emptyset, &\text{otherwise}
\end{cases}$$}\end{definition}
\begin{theorem} \label{Eigenschaften Approxi}Let $P$ be a program, then for all $\alpha\in\aleph_1$ the following holds: \begin{enumerate}\item $\forall\gamma<\alpha\left(M_\gamma =_{\gamma}M_\alpha\right)$
\item$\bigsqcup_{\gamma<\alpha} M_\gamma\sqsubseteq_{\alpha}T_P(\bigsqcup_{\gamma<\alpha} M_\gamma)$ 
\item $M_\alpha=T^{\aleph_1}_{P,\alpha}(\bigsqcup_{\gamma<\alpha} M_{\gamma})$
\item$M_\alpha\sqsubseteq_{\alpha+1}T_P(M_\alpha)$
\end{enumerate}
\end{theorem}\begin{proof} We prove this by induction on $\alpha$. We assume that the theorem holds for all $\beta<\alpha$ (induction hypothesis). We prove that it holds also for $\alpha$.  Using the induction hypothesis, we get that for every $\beta<\alpha$  the following properties hold $\forall\gamma<\beta:M_\gamma =_{\gamma}M_\beta$ as well as $M_\beta\sqsubseteq_{\beta+1}T_P(M_\beta)$.
Then, using Lemma \ref{KonstrHilf1}, we get that $\bigsqcup_{\gamma<\alpha} M_\gamma\sqsubseteq_{\alpha}T_P(\bigsqcup_{\gamma<\alpha} M_\gamma)$  (this is 2.). This together with the above definition imply $M_\alpha=T^{\aleph_1}_{P,\alpha}(\bigsqcup_{\gamma<\alpha} M_{\gamma})$ (this is 3.). Then, using 2. and 3. and Lemma \ref{KonstruktionsHilfe2}, we get that $M_\alpha\sqsubseteq_{\alpha+1}T_P(M_\alpha)$ (this is 4.). It remains to prove the first statement. We know that for all $\gamma<\alpha$ the property $M_\gamma=_\gamma\bigsqcup_{\gamma'<\alpha} M_{\gamma'}{\sqsubseteq_\alpha}^\text{(Lemma \ref{Ketteneigenschaften} \& 2.)} T^{\aleph_1}_{P,\alpha}(\bigsqcup_{\gamma'<\alpha} M_{\gamma'})=^{3.}M_{\alpha}$ holds. Then, using that $\sqsubseteq_\alpha$  is stronger than $=_\gamma$, we get that 1. also holds.\qed\end{proof}
\begin{lemma} \label{null grade}Let $P$ be a program. Then there exists an ordinal $\delta\in\aleph_1$   such that \begin{equation}\forall\gamma\geq\delta:M_\gamma\Vert \mathcal{F}_\gamma=\emptyset\text{ and }M_\gamma\Vert \mathcal{T}_\gamma=\emptyset.\label{deltaPproperty}\end{equation} 
\end{lemma}
\begin{proof} We define the subset $H_B^*$ of the Herbrand base $H_B$ by $H_B^*:=\{A\in H_B;\;\exists\gamma\in\aleph_1:M_\gamma(A)\in\{\mathcal{F}_\gamma,\mathcal{T}_\gamma\}\}.$ Then, using part one of Theorem \ref{Eigenschaften Approxi}, we know that for every $A\in H_B^*$ there is exactly one $\gamma_A$ such that $M_{\gamma_A}(A)\in\{\mathcal{F}_{\gamma_A},\mathcal{T}_{\gamma_A}\}$. Now let us define the function $\zeta$ by $\zeta: H_B^*\rightarrow\aleph_1:A\mapsto\gamma_A.$
We know that $H_B^*$ is countable. This implies that $\zeta(H_B^*)$ is also countable. Then, using Theorem \ref{aleph1regulär}, we know that $\zeta(H_B^*)$ is not cofinal in $\aleph_1$. This obviously implies that there is an ordinal $\delta\in\aleph_1$ such that $\forall A\in H_B^*:\zeta(A)<\delta$. Finally, this ordinal $\delta$ satisfies the property (\ref{deltaPproperty}).\qed\end{proof} \begin{definition}\label{delta_p}\textup{Let $P$ be a program. The lemma above justifies the definition $\delta_P:=\textup{min}\{\delta;\;\forall\gamma\geq\delta:M_\gamma\Vert \mathcal{F}_\gamma=\emptyset\text{ and }M_\gamma\Vert \mathcal{T}_\gamma=\emptyset\}\in\aleph_1$. This ordinal $\delta_P$ is called the \textit{depth} of the program  $P$.}\end{definition}\begin{definition}\label{M_P}\textup{We define the interpretation $M_P$ of a given formula-based logic program $P$ by $$M_P(A):=\begin{cases}M_{\delta_P}(A),&\text{if $\textup{deg}(M_{\delta_P}(A))<\delta_P$}\\0,&\text{otherwise}\end{cases}.$$}\end{definition} \section{Properties of the Interpretation $M_P$}
\begin{proposition}\label{M_p fixpunkt}Let $P$ be a program. The interpretation $M_P$ is a fixed point of $T_P$ (i.e., $T_P(M_{P})=M_{P}$). \end{proposition}
\begin{proof} See Theorem 7.1 in \cite{RondogiannisWadge}.\qed
\end{proof}
\begin{theorem} Let $P$ be a program. The interpretation $M_P$ is a model of $P$.
\end{theorem}
\begin{proof} See Theorem 7.2 in \cite{RondogiannisWadge}.
\qed\end{proof}
\begin{proposition} \label{prop theorem}Let $P$ be a program, $\alpha$ a countable ordinal and  $M$ an arbitrary model of $P$. Then the following holds: $$\forall\beta<\alpha \left(M_\beta=_\beta M\right)\Rightarrow M_\alpha\sqsubseteq_\alpha M$$
\end{proposition} 
\begin{proof} We assume that $\forall\beta<\alpha \left(M_\beta=_\beta M\right)$. Definition \ref{Def vereinigung} implies that \begin{equation}\bigsqcup_{\beta<\alpha} M_{\beta}\sqsubseteq_\alpha M\;\label{(S1)}.\end{equation} Now we prove that the following holds: \begin{equation}T_P(M)\sqsubseteq_\alpha M\;\label{(S2)}\end{equation} Using Lemma \ref{alphamonotonie} and the assumption above, we get that $\forall\beta<\alpha \left(T_P(M_\beta)=_\beta T_P(M)\right)$. This the assumption above and the fourth part of Theorem \ref{Eigenschaften Approxi} imply that $\forall\beta<\alpha:M=_\beta T_P(M)$. But this, together with with the fact that $M$ is a model (i.e.,  $M(A)\geq T_P(M)(A)$ holds for all atoms $A\in H_U$), implies that (\ref{(S2)}) holds.\\We finish the proof by induction on the ordinal $\gamma\in\Omega$. Using Lemma \ref{alphamonotonie} and (\ref{(S2)}), we get that $T^{\gamma}_{P,\alpha}(\bigsqcup_{\beta<\alpha} M_{\beta})\sqsubseteq_\alpha M$ implies $T^{\gamma+1}_{P,\alpha}(\bigsqcup_{\beta<\alpha} M_{\beta})\sqsubseteq_\alpha M$.  Using the first part of Lemma \ref{Ketteneigenschaften}, we get for every limit ordinal $\gamma$ that $\forall\beta<\gamma:T^{\beta}_{P,\alpha}(\bigsqcup_{\beta<\alpha} M_{\beta})\sqsubseteq_\alpha M$ implies $T^{\gamma}_{P,\alpha}(\bigsqcup_{\beta<\alpha} M_{\beta})\sqsubseteq_\alpha M$. Then, using (\ref{(S1)}) and statement 3. of Theorem \ref{Eigenschaften Approxi}, $M_{\alpha}=T^{\aleph_1}_{P,\alpha}(\bigsqcup_{\beta<\alpha} M_{\beta})\sqsubseteq_\alpha M$ holds.\qed\end{proof}
\begin{theorem} \label{Theorem kleinstes Modell} The interpretation $M_P$ of a given program $P$ is the least of all models of $P$ (i.e., for all models $M$ of $P$ the property $M_P\sqsubseteq_\infty M$ holds).
\end{theorem}\begin{proof} Let $M$ be an arbitrary model of $P$. Without loss of generality, we assume that $M\neq M_P$. Then let $\alpha$ be the least ordinal such that $\forall\beta<\alpha \left(M_P=_\beta M\right)$. This implies also $\forall\beta<\alpha \left(M_\beta=_\beta M\right)$. Then, using Proposition \ref{prop theorem}, $M_{P}=_{\alpha}M_\alpha\sqsubseteq_\alpha M$. The choice of $\alpha$ implies that $M_P\neq_\alpha M$. Then we get that $M_P\sqsubset_\alpha M$ and this finally implies $M_P\sqsubseteq_\infty M$.\qed\end{proof}
\begin{corollary}\label{leastfixpoint} Let $P$ be a program. The interpretation $M_P$ is the least of all fixed points of $T_P$.
\end{corollary}
\begin{proof} It is easy to prove that every fixed point of $T_P$ is also a model of $P$. This together with Proposition \ref{M_p fixpunkt} and Theorem \ref{Theorem kleinstes Modell} imply Corollary \ref{leastfixpoint}.\qed\end{proof}
\begin{proposition} There is a countable ordinal $\delta\in\aleph_1$ such that for all programs $P$ of  an arbitrary language  $\mathcal{L}_{n,m,l,(s_i),(r_i)}$ the property $\delta_P<\delta$ holds. Let  $\delta_{\textup{max}}$ be the least ordinal such that the above property holds.
\end{proposition}
\begin{proof} We know that the set of all signatures  $\langle n,m,l,(s_i)_{1\leq i\leq n},(r_i)_{1\leq i\leq m}\rangle$  is countable. Additionally, we know that the set of all programs of a fixed signature is also countable (Remember that a program is a finite set of rules.). This implies that the set of all programs is countable. Then we get that the image of the function from the set of all programs to $\aleph_1$ given by $P\mapsto\delta_P$ is countable. Then, using Theorem \ref{aleph1regulär}, the image of $\delta_{(\cdot)}$ is not cofinal in $\aleph_1$ (i.e., there exists an ordinal $\delta\in\aleph_1$ such that for all programs $P$ the property $\delta_P<\delta$ holds).\qed\end{proof}
\begin{proposition} The ordinal $\delta_{\textup{max}}$ is at least $\omega^\omega$.
\end{proposition}
\begin{proof} Let $n>0$ be a natural number. We consider the program $P_n$ consisting of the following rules (where $G,H$ are predicate symbols, $f$ is a function symbol and c is a constant):
\\\\ $G(x_1,...,x_{n-1},f(x_n))\leftarrow\neg\neg G(x_1,...,x_{n-1},x_n)$
\\For all $k$ provided that $1\leq k\leq n-1$ the rule:
\\ $G(x_1,...,x_{k-1},f(x_k),c,...,c)\leftarrow \exists x_{k+1},...,x_{n}G(x_1,...,x_{k-1},x_k,x_{k+1},...,x_n)$
\\$H\leftarrow\exists x_1,...,x_nG(x_1,...,x_n)$
\\\\This implies that $M_{P_{n}}$ maps $G\left(f^{k_1}(c),...,f^{k_n}(c)\right)$ to $\mathcal{F}_{\sum_{m=1}^{n-1}k_m\omega^{n-m}+k_n\cdot2}$ and $H$ to $\mathcal{F}_{\omega^n}$.\qed
\end{proof}
At the end of this paper we will prove that the 3-valued  interpretation $M_{P,3}$  that results from the infinite-valued model $M_P$  by collapsing all true values to \textit{True} (abbr. $\mathcal{T}$) and all false values to \textit{False} (abbr. $\mathcal{F}$) is also a model in the sense of the following semantics: 
\begin{definition} The semantics of formulas with respect to 3-valued interpretations is defined as in Definition \ref{wfunktion} except that $\llbracket\top \rrbracket^I_h=\mathcal{T}$,  $\llbracket \bot\rrbracket^I_h=\mathcal{F}$ and $$\llbracket \neg(\phi)\rrbracket^I_h=\begin{cases} 
\mathcal{T}, &\text{if }\llbracket\phi\rrbracket^I_h=\mathcal{F}\\
\mathcal{F}, &\text{if }\llbracket\phi\rrbracket^I_h=\mathcal{T}\\
0, &\text{otherwise} \end{cases}.$$  
The Definition \ref{defmodel} is also suitable in the case of 3-valued interpretations. The truth values are ordered as follows: $\mathcal{F}<0<\mathcal{T}$\end{definition}
\begin{proposition} \label{collaps}Let $P$ be a program and let  $\textup{collapse}(\cdot)$ be the function from $W$ to  the set $\{\mathcal{F},0,\mathcal{T}\}$ given by $\mathcal{F}_i\mapsto \mathcal{F}$, $0\mapsto0$ and $\mathcal{T}_i\mapsto \mathcal{T}$. Moreover, let $I$ be an arbitrary interpretation then $\textup{collapse}(I)$ is the 3-valued interpretation given by $\textup{collapse}(I)(A):=\textup{collapse}(I(A))$ for all $A\in H_B$. Then for all formulas $\phi$ and all assignments $h$ the following holds: $$\textup{collapse}(\llbracket \phi\rrbracket^{I}_h)=\llbracket \phi\rrbracket^{\textup{collapse}(I)}_h$$
\end{proposition}
\begin{proof} One can prove  this by induction on the structure of $\phi$ together with Theorem \ref{aleph1regulär}. Due to the page constraints, we present only the most interesting case.\\Let us assume that $\phi=\forall x(\psi)$ and  that the proposition holds for $\psi$. Obviously, the equation $\textup{collapse}(\llbracket \phi\rrbracket^{I}_h)=\text{collapse}(\text{inf}\{\llbracket \psi\rrbracket^{I}_{h[x\mapsto u]};\;u\in H_U\})$ holds. Now we have to consider the following three possible cases, where $I_c:=\text{collapse}(I)\,$:
\\\textit{Case 1:} $\text{inf}\{\llbracket \psi\rrbracket^{I}_{h[x\mapsto u]};\;u\in H_U\}=\mathcal{F}_\alpha$. Then, using Lemma \ref{sup inf existenz annahme}, there must be an $u'\in H_U$ such that $\llbracket \psi\rrbracket^{I}_{h[x\mapsto u']}=\mathcal{F}_\alpha$. This implies, using the assumption, that  $\llbracket \psi\rrbracket^{I_c}_{h[x\mapsto u']}=\mathcal{F}$ and hence $\llbracket\phi\rrbracket_h^{I_c}=\llbracket \forall x(\psi)\rrbracket^{I_c}_{h}=\text{inf}\{\llbracket \psi\rrbracket^{I_c}_{h[x\mapsto u]};\;u\in H_U\}=\mathcal{F}=\textup{collapse}(\mathcal{F}_{\alpha})=\textup{collapse}(\llbracket \phi\rrbracket^{I}_h)$ holds.
\\\textit{Case 2:} $\text{inf}\{\llbracket \psi\rrbracket^{I}_{h[x\mapsto u]};\;u\in H_U\}=\mathcal{T}_\alpha$. Then, using the assumption, we get that $\llbracket \psi\rrbracket^{I_c}_{h[x\mapsto u]}=\mathcal{T}$ for all $u\in H_U$ and hence $\llbracket\phi\rrbracket_h^{I_c}=\llbracket \forall x(\psi)\rrbracket^{I_c}_{h}=\text{inf}\{\llbracket \psi\rrbracket^{I_c}_{h[x\mapsto u]};\;u\in H_U\}=\mathcal{T}=\textup{collapse}(\mathcal{T}_{\alpha})=\textup{collapse}(\llbracket \phi\rrbracket^{I}_h)$ holds.
\\\textit{Case 3:} $\text{inf}\{\llbracket \psi\rrbracket^{I}_{h[x\mapsto u]};\;u\in H_U\}=0$. We know that $H_U$ is a countable set and hence, using Theorem \ref{aleph1regulär}, we get that there must be an $u'\in H_U$ such that $\llbracket \psi\rrbracket^{I}_{h[x\mapsto u']}=0$. The assumption implies that $\llbracket \psi\rrbracket^{I_c}_{h[x\mapsto u']}=0$ and  $0\leq\llbracket \psi\rrbracket^{I_c}_{h[x\mapsto u]}$ for all $u\in H_U$.  Hence we get that $\llbracket\phi\rrbracket_h^{I_c}=\llbracket \forall x(\psi)\rrbracket^{I_c}_{h}=\text{inf}\{\llbracket \psi\rrbracket^{I_c}_{h[x\mapsto u]};\;u\in H_U\}=0=\textup{collapse}(0)=\textup{collapse}(\llbracket \phi\rrbracket^{I}_h)$  holds.\qed
\end{proof}
\begin{proposition} Let $P$ be a formula-based logic program. Then the 3-valued interpretation $M_{P,3}$ is a 3-valued model of $P$.
\end{proposition}
\begin{proof} We assume that $A\leftarrow \phi$ is a rule of $P$. Then, for every assignment $h$, we get that $\llbracket \phi\rrbracket^{M_{P,3}}_h\stackrel{\text{Proposition } \ref{collaps}}{=}\text{collapse}(\llbracket \phi\rrbracket^{M_{P}}_h)\stackrel{\text{Theorem } \ref{Theorem kleinstes Modell}}{\leq}\text{collapse}(\llbracket A\rrbracket^{M_{P}}_h)=\llbracket A\rrbracket^{M_{P,3}}_h$ holds.\qed 
\end{proof}
\begin{remark} The 3-valued model $M_{P,3}$ is not a minimal model in general. Consider the logic program $P=\{P_1\leftarrow\neg\neg P_1\}$. Then the infinite-valued model $M_P$ maps $P_1$ to $0$ and this implies $M_{P,3}(P_1)=0$. But the (2-valued) interpretation $\{\left\langle P_1,\mathcal{F}\right\rangle\}$ is a model of $P$ and it is less than $M_{P,3}$. The ordering on the 3-valued interpretations is introduced in \cite{Przymusinski} page 5. 
\end{remark}
However, Rondogiannis and Wadge prove in \cite{RondogiannisWadge}
that the 3-valued model $M_{P,3}$ of a given normal program $P$ is equal to the 3-valued well-founded model of $P$ and hence, using a result of Przymusinski (Theorem 3.1 of \cite{Przymusinski}), it is a minimal model of $P$. In the context of formula-based logic programs we can prove   Theorem \ref{Minimal Model}. Before we start with the proof we have to consider the following definition and a lemma that plays an important role in the proof of the theorem. 
\begin{definition}\textup{The \textit{negation degree} $\text{deg}_\neg(\phi)$ of a formula $\phi$ is defined recursively on the structure of $\phi$ as follows:
\begin{enumerate}
\item 
If $\phi$ is an atom, then deg$_\neg(\phi):=0$.
\item If $\phi=\psi_1\circ \psi_2$, then deg$_\neg(\phi):=\text{max}\{\text{deg}_\neg(\psi_1), \text{deg}_\neg(\psi_2)\}$. ($\circ\in\{\vee,\wedge\}$)
\item If $\phi=\boxdot x(\psi)$, then deg$_\neg(\phi):=\text{deg}_\neg(\psi)$. ($\boxdot\in\{\exists,\forall\}$)
\end{enumerate}
}\end{definition}
\begin{lemma}\label{LemmaMinimalModel} Let $I$ be an interpretation and $\gamma,\zeta\in\aleph_1$ such that for all $A\in H_B$ the following holds:$$I(A)\in[\mathcal{F}_0,\mathcal{F}_\gamma]\cup\{0\}\cup[\mathcal{T}_\zeta,\mathcal{T}_0]$$ Then for all formulas $\phi$ such that $\textup{deg}_\neg(\phi)\leq1$ and all variable assignments $h$ the following holds:$$\llbracket \phi\rrbracket^I_h\in\begin{cases}[\mathcal{F}_0,\mathcal{F}_\gamma]\cup\{0\}\cup[\mathcal{T}_\zeta,\mathcal{T}_0], & \text{if }\textup{deg}_\neg(\phi)=0 \\
[\mathcal{F}_0,\mathcal{F}_{\textup{max}\{\gamma,\zeta+1\}}]\cup\{0\}\cup[\mathcal{T}_{\textup{max}\{\gamma+1,\zeta\}},\mathcal{T}_0], & \text{otherwise} \\
\end{cases}$$
\end{lemma}
\begin{proof} We prove this by induction on the structure of $\phi$.
\\\textit{Case 1:} $\phi$ is an atom. Obviously, if $\phi=\bot$ or $\phi=\top$, then the lemma holds. Otherwise, there is a ground instance $A\in H_B$ of $\phi$ such that $\llbracket \phi\rrbracket^I_h=I(A)$ and the lemma also holds in this case.\\\textit{Case 2:} $\phi=\psi_1\vee \psi_2$ and the lemma holds for $\psi_1$ and $\psi_2.$ There is an $i\in\{1,2\}$ such that $\llbracket \psi_1\rrbracket^I_h\leq\llbracket \psi_i\rrbracket^I_h$ and $\llbracket \psi_2\rrbracket^I_h\leq\llbracket \psi_i\rrbracket^I_h$. Then, using $\text{deg}_\neg(\psi_i)\leq \text{deg}_\neg(\phi)$ and $\llbracket \phi\rrbracket^I_h=\llbracket \psi_i\rrbracket^I_h$, we get that the lemma also holds for $\phi$.
\\\textit{Case 3:} $\phi=\exists x(\psi)$ and the lemma holds for $\psi$. Then, using Definition \ref{wfunktion}, we get that $\llbracket \phi\rrbracket^I_h=\text{sup}\{\llbracket \psi\rrbracket^I_{h[x\mapsto u]};\;u\in H_U\}$. Let us assume that $\text{deg}_\neg(\phi)=0$. Then, using the assumption of this case and $\text{deg}_\neg(\psi)\leq\text{deg}_\neg(\phi)$, we get that \begin{equation}\text{$\llbracket \psi\rrbracket^I_{h[x\mapsto u]}\in[\mathcal{F}_0,\mathcal{F}_\gamma]\cup\{0\}\cup[\mathcal{T}_\zeta,\mathcal{T}_0]$ for all $u\in H_U$.}\label{lemmaminimaleq 1}\end{equation} This implies that the values $\mathcal{F}_\alpha$ and $\mathcal{T}_\beta$ cannot be least upper bounds (for all $\alpha>\gamma$ and for all $\beta>\zeta$). For instance, assume that $\beta>\zeta$ and $\mathcal{T}_\beta$ is a least upper bound. Then, using statement (\ref{lemmaminimaleq 1}), we get that $0$ must be an upper bound, and hence this contradicts the assumption that $\mathcal{T}_\beta$ is the least upper bound. This implies $\llbracket \phi\rrbracket^I_h=\text{sup}\{\llbracket \psi\rrbracket^I_{h[x\mapsto u]};\;u\in H_U\}\in[\mathcal{F}_0,\mathcal{F}_\gamma]\cup\{0\}\cup[\mathcal{T}_\zeta,\mathcal{T}_0]$ and the lemma holds for $\phi$. Now let us assume that $\text{deg}_\neg(\phi)=1$. This implies that $\llbracket \psi\rrbracket^I_{h[x\mapsto u]}\in[\mathcal{F}_0,\mathcal{F}_{\textup{max}\{\gamma,\zeta+1\}}]\cup\{0\}\cup[\mathcal{T}_{\textup{max}\{\gamma+1,\zeta\}},\mathcal{T}_0]$ for all $u\in H_U$. Then, using the same argumentation as above, we get that  $\llbracket \phi\rrbracket^I_h=\text{sup}\{\llbracket \psi\rrbracket^I_{h[x\mapsto u]};\;u\in H_U\}\in[\mathcal{F}_0,\mathcal{F}_{\textup{max}\{\gamma,\zeta+1\}}]\cup\{0\}\cup[\mathcal{T}_{\textup{max}\{\gamma+1,\zeta\}},\mathcal{T}_0]$ and hence the lemma holds for $\phi$.
\\\textit{Case 4:} $\phi=\neg(\psi)$ and the lemma holds for $\psi$. This implies that $\text{deg}_\neg(\psi)=0$, and hence $\llbracket \psi\rrbracket^I_{h}\in[\mathcal{F}_0,\mathcal{F}_\gamma]\cup\{0\}\cup[\mathcal{T}_\zeta,\mathcal{T}_0]$. If $\llbracket \psi\rrbracket^I_{h}\in[\mathcal{F}_0,\mathcal{F}_\gamma]$, then $\llbracket \neg \psi\rrbracket^I_{h}\in[\mathcal{T}_{\gamma+1},\mathcal{T}_0]$.  If $\llbracket \psi\rrbracket^I_{h}\in\{0\}$, then $\llbracket \neg \psi\rrbracket^I_{h}\in\{0\}$.  If $\llbracket \psi\rrbracket^I_{h}\in[\mathcal{T}_{\zeta},\mathcal{T}_0]$, then $\llbracket \neg \psi\rrbracket^I_{h}\in[\mathcal{F}_{0},\mathcal{F}_{\zeta+1}]$. Hence, the lemma holds also for $\phi$.
\\We omit the case $\phi=\psi_1\wedge \psi_2$ (resp. $\phi=\forall x(\psi)\,)$, since it is similar to Case 2 (resp. Case 3).\qed     
\end{proof}
\begin{theorem}\label{Minimal Model} Let $P$ be a formula-based program such that for every rule $A\leftarrow \phi$ in $P$ the property $\textup{deg}_\neg(\phi)\leq 1$ holds. Then the 3-valued model $M_{P,3}$ of the program $P$ is a minimal 3-valued model.
\end{theorem}
\begin{proof} Let $N_3$ be an arbitrary $3$-valued model of the program $P$, such that $N_3$ is smaller or equal to $M_3$. This is equivalent to \begin{equation}\label{MinimalEquation1}\text{ $M_{P,3}\Vert\mathcal{F}\subseteq N_3\Vert \mathcal{F}$ and $N_3\Vert\mathcal{T}\subseteq M_{P,3}\Vert \mathcal{T}$.}\end{equation} Now we have to prove that $N_3$ is equal to $M_{P,3}$. Note that this holds if and only if both equations $M_{P,3}\Vert\mathcal{F}=N_3\Vert \mathcal{F}$ and $N_3\Vert\mathcal{T}=M_{P,3}\Vert \mathcal{T}$ hold. \\Firstly, we prove that $N_3\Vert\mathcal{T}=M_{P,3}\Vert \mathcal{T}$ by contradiction. We assume that \begin{equation}M_{P,3}\Vert \mathcal{T}\setminus N_{3}\Vert \mathcal{T}\neq\emptyset.\label{MinimalEquation2}\end{equation} We know that $M_{P,3}\Vert\mathcal{T}=\bigcup_{\alpha\in\aleph_1}M_P\Vert\mathcal{T}_\alpha$ and hence, using (\ref{MinimalEquation2}), there must be at least one ordinal $\alpha\in\aleph_1$ such that $M_{P}\Vert \mathcal{T_\alpha}\setminus N_{3}\Vert \mathcal{T}\neq\emptyset$.  This justifies the definition $\alpha_{\text{min}}:=\text{min}\{\alpha\in\aleph_1; \;M_{P}\Vert \mathcal{T_\alpha}\setminus N_{3}\Vert \mathcal{T}\neq\emptyset\}.$
Using Theorem \ref{Eigenschaften Approxi} we get that $M_{\alpha_\text{min}}=T_{P,\alpha_{\text{min}}}^{\aleph_1}(\bigsqcup_{\beta<\alpha_{\text{min}}} M_{\beta})$. To improve readability we define $J:=\bigsqcup_{\beta<\alpha_{\text{min}}} M_{\beta}$. It is obviously that $\alpha_\text{min}<\delta_P$, and hence Definition \ref{M_P}, Theorem \ref{Eigenschaften Approxi}, and Definition \ref{limes aleph_1} imply $M_{P}\Vert \mathcal{T_{\alpha_\text{min}}}=M_{\delta_P}\Vert \mathcal{T_{\alpha_\text{min}}}=M_{\alpha_\text{min}}\Vert \mathcal{T_{\alpha_\text{min}}}=\bigcup_{\gamma\in\aleph_1} T^\gamma_{P,\alpha_\text{min}}(J)\|\mathcal{T}_{\alpha_\text{min}}$. This and the definition of $\alpha_{\text{min}}$ justify the definition $\gamma_{\text{min}}:=\text{min}\{\gamma\in\aleph_1; \;T^\gamma_{P,\alpha_\text{min}}(J)\|\mathcal{T}_{\alpha_\text{min}}\setminus N_{3}\Vert \mathcal{T}\neq\emptyset\}.$ From Definition \ref{Def vereinigung} and Definition \ref{limes aleph_1} we infer that $0<\gamma_\text{min}$ and $\gamma_\text{min}$ is not an infinite limit ordinal, hence $\gamma_\text{min}$ is a successor ordinal. We assume that $\gamma_\text{min}=\gamma_\text{min}^-+1$. Then, using the definition of $\alpha_\text{min}$ and $\gamma_\text{min}$, we get that $T^{\gamma_\text{min}-1}_{P,\alpha_\text{min}}(J)\|\mathcal{T}_{\zeta}\subseteq N_{3}\Vert\mathcal{T}$ for all $\zeta\leq\alpha_\text{min}$. Using statement (\ref{MinimalEquation1}) we infer that $T^{\gamma_\text{min}-1}_{P,\alpha_\text{min}}(J)\|\mathcal{F}_{\zeta}\subseteq N_{3}\Vert\mathcal{F}$ for all $\zeta<\alpha_\text{min}$. Hence, the following definition of the infinite-valued interpretation $N$ is well-defined.$$N(A):=\begin{cases}\mathcal{F}_{\zeta}, & \text{if }\zeta<\alpha_\text{min} \;\&\;A\in T^{\gamma_\text{min}-1}_{P,\alpha_\text{min}}(J)\|\mathcal{F}_{\zeta} \\
\mathcal{F}_{\alpha_\text{min}}, & \text{if }A\in T^{\gamma_\text{min}-1}_{P,\alpha_\text{min}}(J)\|\mathcal{F}_{\alpha_\text{min}}\cap N_3\Vert\mathcal{F} \\
\mathcal{F}_{\alpha_\text{min}+1}, & \text{if }A\in N_3\Vert\mathcal{F}\setminus \bigcup_{\zeta\leq\alpha_\text{min}} T^{\gamma_\text{min}-1}_{P,\alpha_\text{min}}(J)\|\mathcal{F}_{\zeta}\\
\mathcal{T}_{\zeta}, & \text{if }\zeta\leq\alpha_\text{min} \;\&\;A\in T^{\gamma_\text{min}-1}_{P,\alpha_\text{min}}(J)\|\mathcal{T}_{\zeta}  \\
\mathcal{T}_{\alpha_\text{min}+1}, & \text{if }A\in N_3\Vert\mathcal{T}\setminus \bigcup_{\zeta\leq\alpha_\text{min}} T^{\gamma_\text{min}-1}_{P,\alpha_\text{min}}(J)\|\mathcal{T}_{\zeta} \\
0, & \text{otherwise} \\
\end{cases} \;\;\;\;\;\;(\text{for all }A\in H_B)$$               
It is easy to see that \begin{equation}T^{\gamma_\text{min}-1}_{P,\alpha_\text{min}}(J)\sqsubseteq_{\alpha_\text{min}}N\text{ and that }N_3=\text{collapse}(N).\label{MinimalEquation3}\end{equation} Since $T^{\gamma_\text{min}}_{P,\alpha_\text{min}}(J)\|\mathcal{T}_{\alpha_\text{min}}\setminus N_{3}\Vert \mathcal{T}$ is not empty, we can pick an $A$ that is contained in this set. Then, together with Definition \ref{def immediate consequence}, we get that $\mathcal{T}_{\alpha_\text{min}}=T^{\gamma_\text{min}}_{P,\alpha_\text{min}}(J)(A)=T_P(T^{\gamma_\text{min}-1}_{P,\alpha_\text{min}}(J))(A)=\text{sup}\{\llbracket \phi\rrbracket_{I};\;A\leftarrow \phi\in P_G\}$, where $I:=T^{\gamma_\text{min}-1}_{P,\alpha_\text{min}}(J)$. Hence, using Lemma \ref{sup inf existenz annahme}, we can pick a rule $A\leftarrow \phi\in P_G$ such that $\llbracket \phi\rrbracket_I=\mathcal{T}_{\alpha_\text{min}}$. Then, using statement (\ref{MinimalEquation3}), Theorem \ref{Fortsetzungssatz1}, and Proposition \ref{collaps}, we get that $\llbracket \phi\rrbracket_N=\mathcal{T}_{\alpha_\text{min}}$ and $\llbracket \phi\rrbracket_{N_3}=\llbracket \phi\rrbracket_{\text{collapse}(N)}=\text{collapse}(\llbracket \phi\rrbracket_N)=\mathcal{T}$. Lastly, the fact that $N_3$ is a model and $A\leftarrow \phi$ is a ground instance of $P$  imply that $N_3(A)=\mathcal{T}$. But this is a contradiction because we have chosen $A$ to be not contained in $N_{3}\Vert \mathcal{T}$. Hence, statement (\ref{MinimalEquation2}) must be wrong (i.e., $M_{P,3}\Vert\mathcal{T}=N_3\Vert \mathcal{T}$).
\\Secondly, we show that $M_{P,3}\Vert\mathcal{F}=N_3\Vert \mathcal{F}$. Definition \ref{M_P} implies that $M_{P,3}\Vert\mathcal{F}=\bigcup_{\zeta<\delta_P}M_{\delta_P}\Vert\mathcal{F}_\zeta$ and $M_{P,3}\Vert\mathcal{T}=\bigcup_{\zeta<\delta_P}M_{\delta_P}\Vert\mathcal{T}_\zeta$. Then, using \eqref{MinimalEquation1} and the result of the first part of this proof, we get that $\bigcup_{\zeta<\delta_P}M_{\delta_P}\Vert\mathcal{F}_\zeta\subseteq N_3\Vert\mathcal{F}$ and $\bigcup_{\zeta<\delta_P}M_{\delta_P}\Vert\mathcal{T}_\zeta=N_3\Vert\mathcal{T}$. Hence, the following definition of the infinite-valued interpretation $N$ is well-defined and $N_3=\text{collapse}(N)$.
$$N(A):=\begin{cases}\mathcal{F}_{\zeta}, & \text{if }\zeta<\delta_P \;\&\;A\in M_{\delta_P}\|\mathcal{F}_{\zeta}\\
\mathcal{F}_{\delta_P+1}, & \text{if }A\in N_3\Vert\mathcal{F}\setminus M_{P,3}\Vert\mathcal{F}\\
\mathcal{T}_{\zeta}, & \text{if }\zeta<\delta_P \;\&\;A\in M_{\delta_P}\|\mathcal{T}_{\zeta}\\
0, & \text{otherwise}\\
\end{cases}\;\;\;\;\;\;(\text{for all }A\in H_B)$$   
Now we are going to prove by transfinite induction on $\zeta\in\aleph_1$ that $T^{\zeta}_{P,\delta_P+1}(M_{\delta_{P}})\sqsubseteq_{\delta_P+1} N$. Obviously, $T^{\zeta}_{P,\delta_P+1}(M_{\delta_{P}})=_{\delta_P} N$ for all $\zeta\in\aleph_1$.
The Definition of $N$, Definition \ref{delta_p}, and Theorem \ref{Eigenschaften Approxi} imply $N\Vert\mathcal{T}_{\delta_P+1}=\emptyset=M_{\delta_P+1}\Vert\mathcal{T}_{\delta_p+1}=T^{\aleph_1}_{P,\delta_P+1}(M_{\delta_{P}})\Vert\mathcal{T}_{\delta_P+1}= \bigcup_{\gamma<\aleph_1} T^{\gamma}_{P,\delta_P+1}(M_{\delta_P})\|\mathcal{T}_{\delta_P+1}$. Hence, $T^{\zeta}_{P,\delta_P+1}(M_{\delta_P})\|\mathcal{T}_{\delta_P+1}\subseteq N\Vert\mathcal{T}_{\delta_P+1}$ for all $\zeta\in\aleph_1$. It remains to show that $N\Vert\mathcal{F}_{\delta_P+1}\subseteq T^{\zeta}_{P,\delta_P+1}(M_{\delta_P})\|\mathcal{F}_{\delta_P+1}$ for all $\zeta\in\aleph_1$. 
\\\textit{Case 1:} $\zeta=0$. It is easy to prove (using Theorem 4, the result of the first part of this proof, and $N_3\Vert\mathcal{F}\cap N_3\Vert\mathcal{T}=\emptyset$)  that $M_{\delta_P}\Vert \mathcal{F}_{\delta_P+1}=H_B\setminus(M_{P,3}\Vert\mathcal{F}\cup M_{P,3}\Vert\mathcal{T})\supseteq N_3\Vert\mathcal{F}\setminus M_{P,3}\Vert\mathcal{F}=N\Vert\mathcal{F}_{\delta_P+1}$.
\\\textit{Case 2:} $\zeta$ is a successor ordinal and $T^{\zeta-1}_{P,\delta_P+1}(M_{\delta_P})\sqsubseteq_{\delta_P+1} N$. Then, using Definition \ref{limes aleph_1} and Lemma \ref{Ketteneigenschaften}, we get that \begin{equation}T_P(T^{\zeta-1}_{P,\delta_P+1}(M_{\delta_P}))=T^{\zeta}_{P,\delta_P+1}(M_{\delta_P}) \label{minimaleq4}\end{equation} and \begin{equation} T^{\zeta}_{P,\delta_P+1}(M_{\delta_P})\Vert\mathcal{F}_{\delta_p+1}\subseteq T^{\zeta-1}_{P,\delta_P+1}(M_{\delta_P})\Vert\mathcal{F}_{\delta_P+1}.\label{minimaleq5}\end{equation}
We will prove that $T^{\zeta-1}_{P,\delta_P+1}(M_{\delta_P})\Vert\mathcal{F}_{\delta_P+1}\setminus T^{\zeta}_{P,\delta_P+1}(M_{\delta_P})\Vert\mathcal{F}_{\delta_p+1}$ and $N\Vert\mathcal{F}_{\delta_P+1}$ are disjoint. This, using  $T^{\zeta-1}_{P,\delta_P+1}(M_{\delta_P})\sqsubseteq_{\delta_P+1} N$ and statement (\ref{minimaleq5}), implies that $N\Vert\mathcal{F}_{\delta_P+1}\subseteq T^{\zeta}_{P,\delta_P+1}(M_{\delta_P})\|\mathcal{F}_{\delta_P+1}$ and we have proved this case. Therefore, we choose an arbitrary  $A\in T^{\zeta-1}_{P,\delta_P+1}(M_{\delta_P})\Vert\mathcal{F}_{\delta_P+1}\setminus T^{\zeta}_{P,\delta_P+1}(M_{\delta_P})\Vert\mathcal{F}_{\delta_p+1}$. Hence, using Lemma \ref{Ketteneigenschaften}, we get that $\mathcal{F}_{\delta_P+1}<T^{\zeta}_{P,\delta_P+1}(M_{\delta_P})(A)$. This, together with (\ref{minimaleq4}) and Definition \ref{def immediate consequence}, implies that there must be a rule $A\leftarrow \phi\in P_G$ such that $\mathcal{F}_{\delta_P+1}<\llbracket \phi\rrbracket_{I}$, where $I$ is given by $I:=T^{\zeta-1}_{P,\delta_P+1}(M_{\delta_P})$. Then, using the assumption  $I\sqsubseteq_{\delta_P+1} N$ and Theorem \ref{Fortsetzungssatz1}, we get that $\mathcal{F}_{\delta_P+1}<\llbracket \phi\rrbracket_{N}$. We know that for all atoms $C\in H_B$ the image $N(C)$ is an element of $[F_0,F_{\delta_P+1}]\cup\{0\}\cup[T_{\delta_P},T_0]$. Then Lemma \ref{LemmaMinimalModel} and the fact that $\deg_\neg(\phi)\leq 1$ imply $0\leq\llbracket \phi\rrbracket_{N}$. Hence, using Proposition \ref{collaps}, $N_3=\text{collapse}(N)$ and $N_3$ is a model of $P$, we get that $0\leq \llbracket \phi\rrbracket_{N_3}\leq N_3(A)$. Finally, this implies $A\notin N_3\Vert\mathcal{F}\supseteq N_3\Vert\mathcal{F}\setminus M_{P,3}\Vert\mathcal{F}=N\Vert\mathcal{F}_{\delta_p+1}$.\\\textit{Case 3:}  $\zeta>0$ is a limit ordinal and $T^{\gamma}_{P,\delta_P+1}(M_{\delta_P})\sqsubseteq_{\delta_P+1} N$ for all $\gamma<\zeta$. This implies $N\Vert\mathcal{F}_{\delta_P+1}\subseteq T^{\gamma}_{P,\delta_P+1}(M_{\delta_P})\|\mathcal{F}_{\delta_P+1}$ for all $\gamma<\zeta$. Hence, using Definition \ref{limes aleph_1}, we get that $T^{\zeta}_{P,\delta_P+1}(M_{\delta_P})\|\mathcal{F}_{\delta_P+1}=\bigcap_{\gamma\in\zeta} T^\gamma_{P,\delta_P+1}(M_{\delta_P})\|\mathcal{F}_{\delta_P+1}\supseteq N\Vert\mathcal{F}_{\delta_P+1}$.
\\The above transfinite induction shows that $N\Vert\mathcal{F}_{\delta_P+1}\subseteq\bigcap_{\zeta\in\aleph_1} T^{\zeta}_{P,\delta_P+1}(M_{\delta_P})\|\mathcal{F}_{\delta_P+1}$. Then, using that $M_{\delta_P+1}\Vert\mathcal{F}_{\delta_P+1}=\emptyset$ and $M_{\delta_P+1}\Vert\mathcal{F}_{\delta_P+1}=\bigcap_{\zeta\in\aleph_1} T^{\zeta}_{P,\delta_P+1}(M_{\delta_P})\|\mathcal{F}_{\delta_P+1}$, we get that $\emptyset=N\Vert\mathcal{F}_{\delta_P+1}=N_3\Vert\mathcal{F}\setminus M_{P,3}\Vert\mathcal{F}$ (see definition of $N$ above). Last of all, using the assumption (\ref{MinimalEquation1}), we get that  $M_{P,3}\Vert\mathcal{F}=N_3\Vert \mathcal{F}$.\qed         
\end{proof}    
\section{Summary and Future Work}
We have shown that every formula-based logic program $P$ has a least infinite-valued model $M_P$ with respect to the ordering $\sqsubseteq_\infty$ given  on the set of all infinite-valued interpretations. We have presented how to construct the model $M_P$ with the help of the immediate consequence operator $T_P$ and have shown that $M_P$ is also the least of all fixed points of the operator $T_P$. Moreover, we have considered the 3-valued interpretation $M_{P,3}$ and have proven that it is a 3-valued model of the program $P$. Furthermore, we have observed a  restricted class of formula-based programs such that the associated 3-valued models are even minimal models.\\There are some aspects of this paper that we feel should be further investigated. Firstly, we believe that the main results of this work also hold in Zermelo-Fraenkel axiomatic set theory without the Axiom of Choice (ZF). For instance, we could use the class of all ordinals $\Omega$ instead of the cardinal $\aleph_1$ in Theorem \ref{Fortsetzungssatz2}. Secondly, we have proven that the ordinal $\delta_\text{max}$ is at least $\omega^\omega$, but on the other hand we do not know a program $P$ such that $\omega^\omega<\delta_P$. So, one could assume that $\delta_\textup{max}=\omega^\omega$.  Thirdly, the negation-as-failure rule is sound for $M_P$ (respectively, $M_{P,3}$) when we are dealing with a normal program $P$. Within the context of formula-based programs we think it would be fruitful to investigate the rule of definitional reflection presented in \cite{Schroeder-Heister} instead of negation-as-failure. Lastly, we believe that the presented theory can be useful in the areas of databases and data mining. We are looking forward to collaborate with research groups specializing in these areas. 
\\\\\textbf{Acknowledgements.} This work has been financed through a grant made available by the Carl Zeiss Foundation. The author is grateful to Prof. Dr. Peter Schroeder-Heister, Hans-Joerg Ant, M. Comp. Sc., and three anonymous reviewers for helpful comments and suggestions. The final preparation of the manuscript was supported by DFG grant Schr275/16-1.

\end{document}